\documentclass{article}
\usepackage{ijcai24}
\usepackage{times}
\usepackage{soul}
\usepackage{url}
\usepackage[hidelinks]{hyperref}
\usepackage[utf8]{inputenc}
\usepackage[small]{caption}
\usepackage{graphicx}
\usepackage{amsmath}
\usepackage{amsthm}
\usepackage{booktabs}
\usepackage{algorithm}
\usepackage{algorithmic}
\usepackage[switch]{lineno}
\usepackage{bbm}
\usepackage{bm}
\usepackage{natbib}
\usepackage{amssymb}
\usepackage{xcolor}
\usepackage{mathtools}

\newtheorem{theorem}{Theorem}
\newtheorem{lemma}{Lemma}
\newtheorem{define}{Definition}
\newtheorem{corollary}{Corollary}
\title{Stackelberg vs. Nash in the Lottery Colonel Blotto Game}

\author{
Yan Liu$^1$
\and
Bonan Ni$^2$
\and
Weiran Shen$^{1}$
\and
Zihe Wang$^1$
\and
Jie Zhang$^3$
\affiliations
$^1$Renmin University of China \\
$^2$Tsinghua University\\
$^3$University of Bath\\
\emails
liuyan5816@ruc.edu.cn,
bricksern@gmail.com,
shenweiran@ruc.edu.cn,
wang.zihe@ruc.edu.cn,
jz2558@bath.ac.uk
}

\begin{document}
\maketitle

\begin{abstract}
Resource competition problems are often modeled using Colonel Blotto games, where players take simultaneous actions.
However, many real-world scenarios involve sequential decision-making rather than simultaneous moves.

To model these dynamics, we represent the Lottery Colonel Blotto game as a Stackelberg game, in which one player, the leader, commits to a strategy first, and the other player, the follower, responds.
We derive the Stackelberg equilibrium for this game, formulating the leader's strategy as a bi-level optimization problem.

To solve this, we develop a constructive method based on iterative game reductions, which allows us to efficiently compute the leader's optimal commitment strategy in polynomial time.
Additionally, we identify the conditions under which the Stackelberg equilibrium coincides with the Nash equilibrium.
Specifically, this occurs when the budget ratio between the leader and the follower equals a certain threshold, which we can calculate in closed form.
In some instances, we observe that when the leader's budget exceeds this threshold, both players achieve higher utilities in the Stackelberg equilibrium compared to the Nash equilibrium.
Lastly, we show that, in the best case, the leader can achieve an infinite utility improvement by making an optimal first move compared to the Nash equilibrium.
\end{abstract}

\section{Introduction}
Competing for resources with limited budgets in strategic settings has diverse and impactful applications. 
These scenarios arise in areas such as electoral competition, security, crowdsourcing and recommendation systems \citep{Behnezhad-2018,Pagey-2023,Haggiag-2022,Anwar-2024}.
Many of these applications are modeled by Colonel Blotto games \citep{Borel-1921}, which have garnered significant attention from researchers across disciplines, including computer science, economics, and sociology \citep{Kohli-2012,Jackson-2015,Fu-2019}.

Colonel Blotto games have numerous variants, where players allocate limited resources across multiple battlefields.
Resources can be discrete, like troops, or continuous, like bid prices in auctions.
Player utility on each battlefield depends on the invested resources and the battlefield's outcome function, which can follow either a \emph{winner-takes-all} rule or a \emph{proportional} rule.
The latter, known as the \emph{Lottery Colonel Blotto game}, is the focus of this paper.

Colonel Blotto games are highly versatile, and most existing research treats them as normal-form games, emphasizing equilibrium existence and computation in restricted settings \citep{Roberson-2006,Rafael-2012,Macdonell-2015,Perchet-2022}.
However, many real-world scenarios involve sequential decision-making rather than simultaneous play.
For example, in marketplaces such as advertising auctions, e-commerce, and cloud services, large enterprises often act as leaders, committing to budget distribution strategies that smaller retailers observe and respond to.
These dynamics highlight the importance of studying Colonel Blotto games in sequential settings, like Stackelberg games, where one player commits to a strategy before the other makes its decision \citep{Hicks-1935,Schelling-1960}.
In this paper, we model the Lottery Colonel Blotto game as a Stackelberg game.

The Stackelberg game can be viewed as a sequential game.
Many previous studies have focused on scenarios involving an attacker as the follower and a defender as the leader, both operating under cost constraints rather than budget limitations.
In these models, each player determines their effort to maximize their individual payoffs \citep{Zhuang-2007,Cavusoglu-2008,Hausken-2008,Hausken-2012}.
However, \cite{Iliaev-2023} point out that models without budget constraints are easier to solve for Nash equilibria, as they reduce to a single-battlefield model with multiple values.
\cite{Iliaev-2023} also consider budget constraints in their comparison of sequential and simultaneous games.
However, they only impose budget constraints on the follower, while the first mover incurs costs without any budget limitations, and both players have symmetric valuations for the battlefield.
In contrast, we propose a more general model for broader applicability, in which both the leader and the follower have budgets, and the battlefield values are asymmetric.
We also compare the sequential game (the Stackelberg equilibrium) with the simultaneous game (Nash equilibria) to explore the differences in outcomes under these conditions.

To compare sequential and simultaneous games, we analyze the Stackelberg equilibrium and Nash equilibria.
However, obtaining a closed-form representation of Nash equilibria in Colonel Blotto games is a known challenge \citep{Perchet-2022,Li-2022}.
Such representations have only been derived in specific cases, such as when battlefield values are identical \citep{Roberson-2006} or when players have symmetrical budgets \citep{Rafael-2012,Boix-2020}.
We focus on comparing the Stackelberg equilibrium with Nash equilibria, particularly in cases where Nash equilibria can be explicitly solved.
In these solvable instances, we have made some interesting observations.
For the Stackelberg equilibrium, the key challenge is computing the leader's optimal commitment.
Previous work has suggested negative conclusions about this problem.
Specifically, finding an optimal pure strategy for the leader is NP-hard in normal-form games \citep{Conitzer-2006,Korzhyk-2010,Letchford-2010}, as it involves solving a bi-level optimization problem \citep{Renou-2009}.
To address this, we explore new insights into optimal commitment and best-response dynamics, which allow us to compare the Stackelberg equilibrium with Nash equilibria in these games.

\subsection{Our Contribution}
Our contributions are multi-fold:
\begin{itemize}
    \item We offer a novel understanding of the follower's best response using a water-filling approach.
    \item We construct a series of game reductions by splitting battlefields, ensuring that the follower's valuation for each sub-battlefield is uniform.
    We have proven that there exists an injective mapping between players' strategies before and after each reduction of the game.
    Additionally, the players' utilities remain unchanged.
    By analyzing the leader's commitment in the reduced game, we characterize the optimal commitment in the original game.
    Ultimately, we reduce the support of the follower's best response strategy from $2^n - 1$ to $n$, where $n$ is the number of battlefields.
    \item By reformulating the leader's objective and reducing the search space for the follower's best response strategy, we can compute the leader's optimal commitment strategy by a polynomial number of iterations. 
    \item We provide the sufficient and necessary conditions, under which the Stackelberg equilibrium coincides with the Nash equilibria of the simultaneous-move game.
    \item Finally, we address our motivating questions by constructing extreme cases that compare the leader's utility when using the optimal commitment strategy to its utility in Nash equilibria. 
    We observe that this ratio is also dependent on the leader's budget relative to the follower's budget.
    Furthermore, we provide an example illustrating how the follower's utility may still increase when faced with the leader's optimal commitment.
\end{itemize}

\subsection{Related Work}
{\bf Colonel Blotto Games and Their Variants.} There is an extensive body of literature on Colonel Blotto games.

When the outcome function is winner-takes-all, the player who allocates the most resources to a battlefield wins that battlefield.
In this line of work, a pure Nash equilibrium does not always exist, while a mixed Nash equilibrium always does \citep{Roberson-2006,Macdonell-2015}.
Specifically, each player's strategy is a complex joint distribution over an n-dimensional simplex.
\cite{Roberson-2006} provides a closed-form solution for the Nash equilibrium in the case of two players and multiple battlefields with equal value.
\cite{Macdonell-2015} provide a detailed and comprehensive analysis of the Nash equilibrium in the case of two players and two battlefields with different values.
In view of the difficulty in analyzing Nash equilibrium, researchers propose the \textit{General Lotto game}, a well-known variant.
Specifically, each player's strategy is a joint distribution over the n-dimensional simplex, ensuring that the expected allocation of resources does not exceed the budget, rather than strictly staying within the budgets \citep{Hart-2008,Dziubinski-2013,Kovenock-2021}.
The Nash equilibrium solution of the General Lotto game has been obtained \citep{Kovenock-2021}.

When the outcome function is proportional, the probability of a player winning a battlefield depends on the proportion of resources they allocate compared to the total resources allocated by all players.
This type of Colonel Blotto games is also called the \textit{Lottery Colonel Blotto game}.
In this line of work, it is called \textit{symmetric} if all players have equal budgets; otherwise, it is \textit{asymmetric}.
When all battlefields have the same value, and this value is consistent for all players, the game is \textit{homogeneous}.
If the battlefields possess different values, but these values remain the same for every player, the game is \textit{heterogeneous}.
Finally, when the battlefields' values differ, and these values vary for each player, the game is termed a \textit{generalized} game.
The \textit{Lottery Colonel Blotto game} has been proven to possess pure Nash equilibria \citep{Kim-2018}.
\cite{Friedman-1958} shows that the pure Nash equilibrium is unique and exhibits proportionality features in a two-player setup with heterogeneous battlefield values and asymmetric player budgets.
\cite{Duffy-2015} extend Friedman's analysis from two players to multiple players.
\cite{Kim-2018} provide a method to identify all pure Nash equilibria in the case of two players with asymmetric budgets and multiple battlefields with generalized values, and present an example where the pure Nash equilibrium is not unique.
\cite{Kovenock-2019} provide the best response given the strategy of the other player.
However, none of the existing work has provided a closed-form representation of the pure Nash equilibria \citep{Kim-2018,Kovenock-2019}.
In a variant where the outcome function of a battlefield is determined by the Tullock rent-seeking contest success function, the proportional rule is parameterized by a battle-specific discriminatory power and a battle-and-contestant-specific lobbying effectiveness \citep{Xu-2018,Osorio-2013,Li-2022}.
In this variant, the existence of a pure Nash equilibrium is not always guaranteed.
\cite{Xu-2022} provide sufficient conditions for the existence of the pure Nash equilibria, while \cite{Li-2022} analyze the properties of the pure Nash equilibria and provide sufficient conditions for its uniqueness.

\noindent {\bf Sequential Colonel Blotto Games. } There are two types of sequentiality: one is the sequentiality of the actions taken by players, and the other is the sequentiality of the appearance of the battlefields.

In the sequentiality of the actions taken by players, the first player makes a move, followed by the second player.
Many studies have compared sequential games and simultaneous games, in other contexts.
For instance, \cite{Zhuang-2007} identify equilibrium strategies for both attacker and defender in simultaneous and sequential games, although their model does not consider players' budgets.
\cite{Chandan-2020} consider three-stage Colonel Blotto games with stronger and weaker players.
In their model, the weaker player has the option to pre-commit resources to a single battlefield of its choice, and the stronger player can choose whether to allocate resources to win that battlefield.
\cite{Chandan-2022} propose a two-stage General Lotto game, in which one of the players has reserved resources that can be strategically pre-allocated across the battlefields in the first stage and the players then engage by simultaneously allocating their real-time resources against each other.

In the sequentiality of the appearance of the battlefields, all players simultaneously allocate resources in the first battlefield, then allocate resources in the second battlefield, and so on.
\cite{Anbarci-2023} analyze dynamic Blotto contests, where battlefields are presented to players in a predetermined sequential order.
They focus on the sub-game perfect equilibrium, exploring the existence and uniqueness of this solution concept.
\cite{Xie-2022} construct a pure strategy Markov perfect equilibrium (when it exists) and provide closed-form solutions for players’ strategies and winning probabilities.
\cite{Klumpp-2019} explore the strategic allocation of resources in a dynamic setting where winning a majority of battlefields is the goal.
They provide the optimal strategies for both players in sub-game perfect equilibrium.
\cite{Li-2021} reveal that the even-split strategy is robust when players have incomplete information about the other player's resource allocation.

\noindent {\bf Other Variants and Applications.} 
The Colonel Blotto games and their variants, along with the Tullock contest, all-pay auction, and their generalizations, can find wide application in various domains.
These include competition design \citep{Deng-2023}, contest design \citep{Ghosh-2014,Levy-2017,Letina-2023,Dasgupta-1998}, and auctions \citep{Wang-2015,Branzei-2012}.

\section{Preliminaries}
In the Stackelberg model of the Lottery Colonel Blotto game, let $a$ and $b$ represent the leader and the follower, respectively.
Both players have limited budget constraints, where $x_a > 0$ and $x_b > 0$.
These players allocate their budgets across $n$ battlefields, represented by the set $[n] = \{1, 2, \dots, n\}$.
For player $i \in \{a, b\}$, let $x_{ij} \geq 0$ denote the budget invested by player $i$ in battlefield $j$.
Throughout this paper, we consider players' pure strategies, which are denoted as $\bm{x}_i = (x_{i1}, x_{i2}, \cdots, x_{in})$.
The strategy set for player $i$ is:
$\bm{X}_i \triangleq \{\bm{x}_i: \sum_{j = 1}^{n} x_{ij} = x_i \text{ and } x_{ij} \geq 0 \}.$

Player $i \in \{a, b\}$ assigns a value to battlefield $j$ as $v_{ij} \in \mathbb{Q}_{>0}$.
Hence, the game can be represented as
$$\mathcal{G} := \langle \{a, b\}, [n], x_a, x_b, (v_{aj})_{j = 1}^{n}, (v_{bj})_{j = 1}^{n} \rangle.$$

Given a strategy profile $(\bm{x}_a, \bm{x}_b)$, the utility of player $i \in \{a, b\}$ on battlefield $j \in [n]$ is defined as the proportion of the budget allocated by player $i$ relative to the total budgets allocated by both players.
That is,
\begin{equation}\label{sec2-equ-uij}
    u_{ij}(x_{ij}, x_{-ij}) = \frac{x_{ij}}{x_{ij} + x_{-ij}} \cdot v_{ij}.
\end{equation}
In the event that both players allocate zero budget to a battlefield $j$, we assume that the follower will win the entire battlefield. 
This assumption is based on the rationale that the follower can achieve this outcome by allocating even an arbitrarily small amount of budget to $j \in [n]$. 
This assumption guarantees the existence of the follower's best response strategy.
Using a linear aggregation function, player $i$'s utility in the game is given by:
\begin{equation}\label{sec2-equ-ui}
    u_i(\bm{x}_a, \bm{x}_b) = \sum_{j = 1}^{n}u_{ij} = \sum_{j = 1}^{n} \frac{x_{ij} \cdot v_{ij}}{x_{ij} + x_{-ij}} , \; \forall i \in \{a, b\}.
\end{equation}

When $x_{aj} > 0$, for all $j \in [n]$, \cite{Kovenock-2019} characterize the other player's best response function as follows.
\begin{lemma}\label{sec2-lemma-br}
    \cite{Kovenock-2019}
    Given the leader's strategy $\bm{x}_a = (x_{aj})_{j = 1}^{n}$, where $x_{aj} > 0$, $\forall j \in [n]$, assume without loss of generality that all battlefields are ordered such that $\frac{v_{b1}}{x_{a1}} \geq \frac{v_{b2}}{x_{a2}} \geq \cdots \geq \frac{v_{bn}}{x_{an}}$.
    The unique optimal budget allocation of the follower, $\bm{x}_b = (x_{b1}, x_{b2}, \cdots, x_{bn})$, to battlefield $j \in [n]$ is characterized as follows:
    \begin{equation}\label{sec2-lem-equ-br}
        \begin{aligned}
            x_{bj} = \begin{cases}
                \frac{\left( x_{aj} v_{bj} \right)^{\frac{1}{2}} \left( x_b + \sum\limits_{j' \in K(\bm{x}_a)} x_{aj'} \right)}{\sum\limits_{j' \in K(\bm{x}_a)}\left( x_{aj'} v_{bj'} \right)^{\frac{1}{2}}} - x_{aj}, & \text{if } j \in K(\bm{x}_a); \\
                0, & \text{if } j \in [n] \backslash K(\bm{x}_a),
            \end{cases}
        \end{aligned}
    \end{equation}
    where $K(\bm{x}_a) = \{ 1, \cdots, k^* \}$ is such that
    \begin{equation}\label{shorteststave}
        \displaystyle
        k^* = \max \left\{ k \in [n]: \frac{v_{bj}}{x_{aj}} > \frac{\left( \sum\limits_{l = 1}^j (x_{al} v_{bl})^{\frac{1}{2}} \right)^2}{\left( x_b + \sum\limits_{l = 1}^j x_{al} \right)^2}, \forall j \leq k \right\}.
    \end{equation}
\end{lemma}

This characterization is a useful tool for our subsequent analysis. 
To use this result to characterize the leader's optimal commitment, we first assume that the leader will allocate a positive budget to every battlefield. 
Following a series of constructions in Section \ref{BRsupport}, we identify the support of the follower's best response strategies and provide a closed-form expression of the leader's optimal commitment strategy in Section \ref{OptCommitment}.
To complete this analysis, we will verify that the leader indeed allocates a positive budget to every battlefield in the optimal commitment strategy. 
We focus on the pure strategy of the leader.

The uniqueness of the follower's best response $\bm{x}_b$ will facilitate the subsequent analysis. 
However, providing a closed-form representation of the leader's optimal commitment strategy remains very challenging. 
This difficulty primarily arises from the uncertainty with regard to which battlefields the follower will abandon, specifically where $x_{bj} = 0$ for $j \in [n] \backslash K(\bm{x}_a)$. 
For simplicity, we denote $\overline{K(\bm{x}_a)} = [n] \backslash K(\bm{x}_a)$ and the support $k^* = |K(\bm{x}_a)|$.
In the following section, we characterize the support of the follower's best response strategy when the leader employs optimal commitment strategy. 
This analysis, in turn, will aid in fully computing the leader's optimal commitment strategy.

\section{The Support of the Follower's Best Response Strategy} \label{BRsupport}
The computation of the leader's optimal commitment strategy can be formulated as a bi-level optimization problem. 
In this scenario, the leader, assuming the follower is a utility maximizer, selects an optimal strategy in the upper-level optimization task, while the follower provides the best response in the lower-level optimization task.
The complexity is compounded by the fact that the support of each player's strategy profile can have up to $2^n - 1$ possible combinations. 
Additionally, it is challenging to determine the exact budget allocation for each battlefield.

In this section, we demonstrate that when the leader employs optimal commitment strategy, the support of the follower's best response strategy is limited to at most $n$ possible combinations. 
This significantly reduces the search space for the leader's optimal commitment strategy.

We achieve this characterization by constructing a series of auxiliary games.
\begin{enumerate}
    \item Consider a game $\mathcal{G}$, with the leader's commitment strategy denoted by $\bm{x}_a$ (which may not be optimal), and the follower's best response denoted by $\bm{x}_b$. 
    If we divide a single battlefield into multiple sub-battlefields in such a way that the leader and follower's valuations of the original battlefield are evenly distributed among these sub-battlefields, the resulting game is denoted by $\mathcal{G}^{(1)}$. 
    By evenly distributing their budgets across these sub-battlefields, the utilities for both the leader and the follower remain unchanged in $\mathcal{G}^{(1)}$.
    The formal statement and its proof are deferred to Lemma 6.
    \item By repeatedly performing this battlefield-splitting process, we create a game in which the follower values all battlefields equally.
    Without loss of generality, we can rename the battlefields so that the leader values them in increasing order. 
    Denote this game by $\mathcal{G}^{(2)}$. 
    According to Lemma 7, in the optimal commitment strategy, the leader should allocate more or at least equal budget to battlefields that they value strictly higher in order to achieve a strictly greater payoff.
    \item Furthermore, we observe that, in the optimal commitment strategy, when the leader values two battlefields equally, they should allocate an equal budget to both battlefields. 
    For a formal statement and proof, see Lemma 8.
    \item Following the above procedure, we obtain a game where $v_{b1} = v_{b2} = \dots$, meaning the follower values all battlefields equally.
    Meanwhile, the leader values the battlefields in increasing order, with some values strictly increasing and others remaining equal. 
    Our final operation is to merge these battlefields that split from the same battlefield.
    In this new game $\mathcal{G}^{(3)}$, the leader and follower's values and budgets across these sub-battlefields are merged into the original battlefield. 
    Their utilities remain unchanged from $\mathcal{G}^{(2)}$. 
    The formal statement and proof are deferred to Lemma 9.
\end{enumerate}

With these constructions and observations, we now provide the main result of this section.
\begin{theorem}\label{theorem-characterization}
    In a Stackelberg game
    $$\mathcal{G} := \langle \{a, b\}, [n], x_a, x_b, (v_{aj})_{j = 1}^{n}, (v_{bj})_{j = 1}^{n} \rangle,$$
    assume without loss of generality that the battlefields are ordered such that $\frac{v_{a1}}{v_{b1}} \leq \frac{v_{a2}}{v_{b2}} \leq \cdots \leq \frac{v_{an}}{v_{bn}}$.
    Then, when the leader employs the optimal commitment strategy $\bm{x}_a$, the support of the follower's best response strategy has at most $n$ possibilities.
    Specifically, $K(\bm{x}_a) \in \{ \{1\}, \{1,2\}, \dots, \{1,2,\dots,n\} \}$.
    Additionally, if for $j, h \in [n]$, $\frac{v_{aj}}{v_{bj}} = \frac{v_{ah}}{v_{bh}}$, then the optimal commitment $\bm{x}_a$ satisfies $\frac{x_{aj}}{x_{ah}} = \frac{v_{aj}}{v_{ah}}$.
\end{theorem}
{\bf Technical Remark:}
The constructions in this section to derive the main result rely on an interpretation different from the one provided by \cite{Kovenock-2019}.
To prove Lemma \ref{sec2-lemma-br}, they represent the second player's best response strategy as the solution to an optimization problem.
By verifying that \eqref{sec2-lem-equ-br} satisfies the Kuhn-Tucker conditions, which are both necessary and sufficient, they conclude that \eqref{sec2-lem-equ-br} is the unique global constrained maximizer of the problem.

In contrast, we interpret the follower's optimization problem as a {\bf water-filling} process. 
Initially, the follower allocates their budget to the battlefields with the highest marginal utility. 
Note that the first and second derivatives of the follower's utility from battlefield $j$ are:
\begin{equation}\label{equation-marginal-utility-follower}
    \begin{aligned}
        &\frac{\partial u_{bj}}{\partial x_{bj}} = \frac{\partial ( \frac{x_{bj} v_{bj}}{x_{aj} + x_{bj}} )}{\partial x_{bj}} = \frac{x_{aj} v_{bj}}{(x_{aj} + x_{bj})^2}, \\
        &\frac{\partial^2 u_{bj}}{\partial x_{bj}^2} = \frac{-2x_{aj}v_{bj}}{(x_{aj} + x_{bj})^3} < 0.
    \end{aligned}
\end{equation}

Hence, the utility $u_{bj}$ is a concave function with respect to $x_{bj}$. 
As the budget allocation to battlefield $j$ increases, its marginal utility decreases until it matches the marginal utility of other battlefields with initially lower marginal utility. 
From this point, the follower distributes its budget across these battlefields, maintaining equal marginal utility, until it eventually decreases to an even lower level. 
Ultimately, the follower exhausts its entire budget. 
The lowest marginal utility is given by $\frac{\left( \sum_{l = 1}^{k^*} (x_{al} v_{bl})^{\frac{1}{2}} \right)^2}{\left( x_b + \sum_{l = 1}^{k^*} x_{al} \right)^2}$, as defined in $k^*$ in \eqref{shorteststave}.

It is essential to consider whether these game-splitting and merging operations, as well as the leader's strategies $\bm{x}_a$, alter the structure of the set $K(\bm{x}_a)$. 
This set, $K(\bm{x}_a)$, represents the battlefields on which the follower allocates positive budgets. 
Essentially, there is a mapping between the leader's commitment set and the follower's strategy set $\bm{x}_b$ across the original game $\mathcal{G}$ and the subsequent games $\mathcal{G}^{(k)}$ for $k = 1, 2, 3$.

\section{The Leader's Optimal Commitment Strategy} \label{OptCommitment}
Theorem \ref{theorem-characterization} characterizes $K(\bm{x}_a)$, the support of the follower's best response strategy when the leader employs its optimal commitment strategy $\bm{x}_a$.
In this section, we formulate the leader's objective as an optimization problem for any given set $K(\bm{x}_a)$.
By solving the optimization problem for all $n$ possible realizations of the set $K(\bm{x}_a)$ as described in Theorem \ref{theorem-characterization}, and comparing the optimal values, we identify the optimal commitment strategy $\bm{x}_a$.

Assume the support of the follower's best response strategy, when the leader employs its optimal commitment strategy, is given by $K$.
Given a specific set $K(\bm{x}_a) = K$, i.e., the set of battlefields in which the follower participates when playing a best response, the leader's optimal commitment strategy can be determined by solving the following optimization problem ({\bf OC}).
\begin{align}
    \max_{\bm{x}_a \in \bm{X}_a} \quad & u_a(\bm{x}_a, \bm{x}_b) = \sum_{j \in K} \frac{x_{aj} \cdot v_{aj}}{x_{aj} + x_{bj}} + \sum_{j \in \overline{K}} v_{aj} & \label{optimal-commitment} \\
    \mbox{s.t.} \quad &\sum_{j = 1}^{n} x_{aj} = x_a, \label{budgetconstraint} &  \\
    &x_{aj} > 0,  \forall j \in [n], &  \label{xaassumption} \\
    &x_{bj} = \frac{(x_{aj} v_{bj})^\frac{1}{2} (x_b + \sum\limits_{j' \in K} x_{aj'})}{\sum_{j' \in K}(x_{aj'} \cdot v_{bj'})^\frac{1}{2}} - x_{aj} > 0, \forall j \in K, \label{bestresponse1}   \\
    &(\frac{v_{bj}}{x_{aj}})^\frac{1}{2} \leq \frac{\sum_{l \in K} (x_{al} v_{bl})^\frac{1}{2}}{x_b + \sum_{l \in K} x_{al}}, \forall j \notin K, \label{noteq-conatrain} \\
    &x_{bj} = 0, \forall j \in \overline{K}. \label{bestresponse2} &
\end{align}
The leader's utility is derived from two components: the battlefield set $K$, where the leader and the follower share the battlefield proportional to their resources allocated, and the set $\overline{K}$, where the leader has completely won the battlefield. 
Equation \eqref{budgetconstraint} represents the leader's budget constraint. 
Inequalities \eqref{xaassumption} assume the leader allocates a positive amount of budget to every battlefield in the optimal commitment strategy. 
Constraints \eqref{bestresponse1} to \eqref{bestresponse2} describe the follower's best response strategy as outlined in Lemma \ref{sec2-lemma-br}.

Next, we characterize the budget that the leader allocates to the battlefields where the follower does not compete. 
This characterization is based on interpreting the follower's best response strategy through a water-filling approach, but from the leader's perspective. 
Consider any battlefield $j \in \overline{K}$.
On one hand, the leader must allocate sufficient budget to these battlefields so that the follower's marginal utility on these battlefields is lower than on other battlefields, making competition unappealing. 
On the other hand, the leader does not need to allocate an excessive amount of resources to force the follower to withdraw. 
Therefore, there exists a minimum threshold of resources that the leader must allocate to these battlefields in $\overline{K}$, which is sufficient but not wasteful. 
This threshold value is provided below.
\begin{lemma}\label{sec4-lemma-leader-threshold}
    In the leader's optimal commitment $\bm{x}_a$, we have $x_{aj} = \frac{v_{bj} \cdot \left( x_b + \sum_{l \in K(\bm{x}_a)} x_{al} \right)^2}{\left( \sum_{l \in K(\bm{x}_a)} (x_{al} v_{bl})^{\frac{1}{2}} \right)^2}, \forall j \in \overline{K(\bm{x}_a)}.$
\end{lemma}

We simplify the optimization problem ({\bf OC}) in three aspects: (i) we add the expression for $x_{aj}, j \in \overline{K}$ as a constraint in ({\bf OC}); (ii) we substitute $x_{bj}$ as outlined in \eqref{bestresponse1} and \eqref{bestresponse2} into the objective function \eqref{optimal-commitment}, thereby removing $x_{bj}$ from the leader's optimization problem; (iii) we note that the second term of \eqref{optimal-commitment} can be removed, as it becomes a constant when the set $\overline{K}$ is fixed.
Denote the objective $\hat{u}_a(\bm{x}_a) = u_a(\bm{x}_a, \bm{x}_b) - \sum_{j \in \overline{K}} v_{aj}$.

Therefore, given the set $K(\bm{x}_a) = K$, the optimization problem ({\bf OC}) can be reformulated as (\text{\bf{OC'}}), as shown below.
\begin{align}
    \max_{\bm{x}_a \in \bm{X}_a} \quad &\hat{u}_a(\bm{x}_a) = \left( \sum_{j \in K} (\frac{x_{aj}}{v_{bj}})^{\frac{1}{2}} v_{aj} \right) \left( \frac{\displaystyle\sum_{j \in K} (x_{aj} v_{bj})^\frac{1}{2}}{x_b + \sum\limits_{j \in K} x_{aj}} \right)  \label{rewrite-optimal-commitment} \\
    \mbox{s.t.} \quad & x_{aj} = \frac{v_{bj} \cdot \left( x_b + \sum_{l \in K} x_{al} \right)^2}{\left( \sum_{l \in K} (x_{al} \cdot v_{bl})^{\frac{1}{2}} \right)^2}, \quad \forall j \in \overline{K}, & \label{xajjinkbar} \\
    &\sum_{j = 1}^{n} x_{aj} = x_a, &  \label{atotalbudget} \\
    &x_{aj} > 0, \quad \forall j \in [n]. &  \nonumber
\end{align}

Solving ({\bf OC'}) is still challenging, as it is not a typical convex programming problem that can be solved in polynomial time.
Additionally, heuristics that approximate the optimal solution do not help us in ultimately comparing the Stackelberg equilibrium with the Nash equilibria. 
To address this challenge, we develop the following characterization of the leader's optimal commitment strategy $x_{aj}$, when $j \in K$.
\begin{lemma}\label{sec4-lemma-parameters}
    Let $\bm{x}_a$ be the leader's optimal commitment. There exist two parameters, $\alpha$ and $\beta$, such that $\frac{v_{aj}}{\sqrt{v_{bj}}} - \alpha \sqrt{v_{bj}} = \sqrt{x_{aj}} \beta$, for all $ j \in K(\bm{x}_a)$.
\end{lemma}

Together with \eqref{xajjinkbar} and \eqref{atotalbudget}, we can eliminate the parameter $\beta$ in Lemma \ref{sec4-lemma-parameters} by establishing an equation through the leader's optimal commitment strategy $\bm{x}_a$. 
By substituting $x_{aj}$'s into \eqref{rewrite-optimal-commitment}, we reformulate the leader's utility as a function of the parameter $\alpha$. 
This transforms the non-convex optimization problem (\textbf{OC'}) into a single-variable function maximization problem. 
Note that $K(\bm{x}_a)$ can be one of the sets $\{ \{1\}, \{1,2\}, \dots, \{1,2,\dots,n\} \}$.
By considering one possible realization of $K(\bm{x}_a)$ at a time and taking the derivative of the objective function with respect to $\alpha$, we can identify $x_{aj}$'s as below. 
We refer to each of these strategies $\bm{x}_a$ that corresponds to a set $K(\bm{x}_a)$ as a \textit{candidate optimal commitment strategy}.

\begin{theorem}\label{sec4-theorem-optimal-commitment}
In a Stackelberg game
$$\mathcal{G} := \langle \{a, b\}, [n], x_a, x_b, (v_{aj})_{j = 1}^{n}, (v_{bj})_{j = 1}^{n} \rangle,$$
for each set $K \in \{\{1\}, \{1, 2\}, \cdots, \{1,2,\cdots,n\}\}$, we can, in $O(n)$ steps, formulate an optimization problem involving a univariate continuous function defined over the union of two half-open intervals. The solution to this problem represents the leader's optimal commitment strategy.
\end{theorem}
Due to the complexity of the expression of the leader's optimal commitment, we provide its specific formula in the proof of Theorem \ref{sec4-theorem-optimal-commitment}.
Since the optimization objective for each $K$ is a univariate continuous function, this allows us to perform simulation experiments.

{\bf The Leader's Optimal Commitment Strategy.} Following Theorem \ref{sec4-theorem-optimal-commitment}, in a Stackelberg game
$\mathcal{G} := \langle \{a, b\}, [n], x_a,$ $ x_b, (v_{aj})_{j = 1}^{n}, (v_{bj})_{j = 1}^{n} \rangle$, we compute the candidate optimal commitment strategy $\bm{x}_a$ for each set $K(\bm{x}_a) \in \{ \{1\}, \{1,2\},$ $\dots, \{1,2,\dots,n\} \}$.
For each candidate strategy $\bm{x}_a$, we then determine the leader's utility.
The strategy $\bm{x}_a$ that yields the highest utility is identified as the leader's optimal commitment strategy.

Recall that the best response function provided by \citep{Kovenock-2019} requires $x_{aj} > 0$, for all $j \in$ $[n]$. 
To conclude this section, we demonstrate that if there exists a battlefield $j$ such that $x_{aj}=0$ in the leader's optimal commitment strategy, as identified through our process, then increasing $x_{aj}$ to an arbitrarily small budget and carefully adjusting resource allocation in other battlefields will enhance the leader's utility.
Therefore, the leader's optimal commitment strategy must satisfy this prerequisite.

\begin{lemma}\label{sec4-lemma-allocate-positive-resources-allbattlefields}
    In the leader's optimal commitment $\bm{x}_a$, for $\forall j \in [n]$, $x_{aj} > 0$.
\end{lemma}

\section{Stackelberg vs. Nash} \label{Stackelberg vs. Nash}
In this section, we investigate the necessary and sufficient conditions under which the Stackelberg equilibrium coincides with the Nash equilibria.

Let's define $\frac{v_{aj}}{v_{bj}}$ as the \textit{relative value ratio} of the two players over battlefield $j$. 
Our first observation is that the relative value ratio can only have two distinct values if a Stackelberg equilibrium is also a Nash equilibrium.

\begin{lemma}\label{sec5-lemma-nece-suff}
   Let $\bm{x}_a$ be the leader's optimal commitment strategy and $\bm{x}_b$ be the follower's best response to $\bm{x}_a$. 
   If $\bm{x}_a$ is also the best response strategy to $\bm{x}_b$, then the cardinality of the set $\{\frac{v_{aj}}{v_{bj}}, \forall j \in [n] \}$ is at most two.
\end{lemma}

To prove this lemma, we apply Lemma \ref{sec2-lemma-br} twice, since $\bm{x}_a$ and $\bm{x}_b$ are best response strategies to each other.
Together with Lemma \ref{sec4-lemma-parameters}, we can establish a quadratic equation whose single variable is $\frac{v_{aj}}{v_{bj}}, \forall j \in [n]$. 
Since a quadratic equation can have at most two distinct roots, the lemma is proved.

Lemma \ref{sec5-lemma-nece-suff} allows us to merge battlefields with identical relative value ratios. This step is essential for proving the main result of this section, as detailed below. 
\begin{figure}[ht]
    \centering
    \includegraphics[width=0.3\columnwidth]{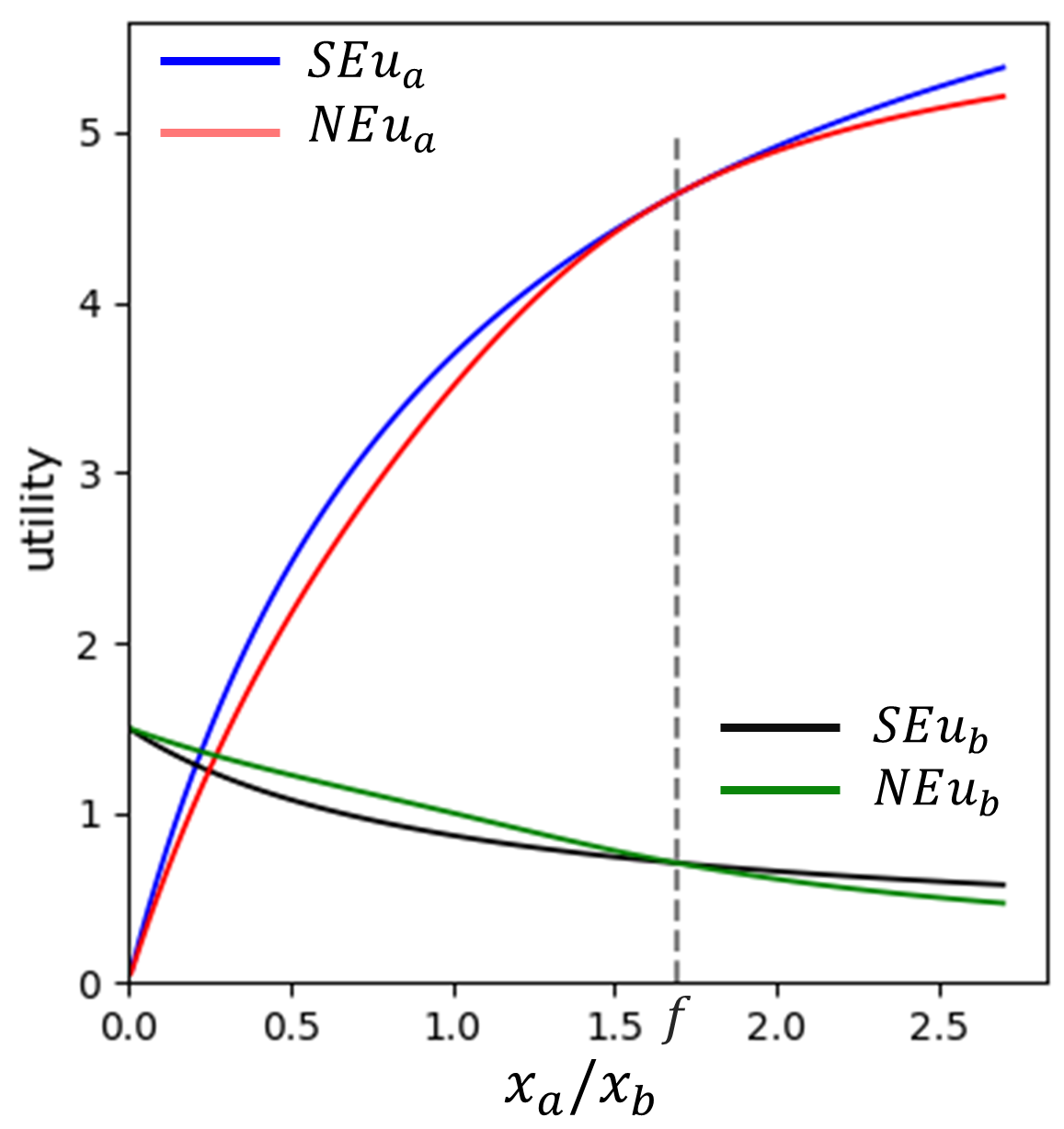}
    \caption{Utility Curves in Two Equilibria. }
    \label{utility:example}
\end{figure}

\begin{theorem}\label{finally-theorem}
    There exists a function $f: \mathbb{R}^4 \rightarrow \mathbb{R}$ such that
    a Stackelberg equilibrium is also a Nash equilibrium if and only if one of the following conditions hold:
    \begin{itemize}
        \item The relative value ratio is consistent across all battlefields. 
        That is, there exists a constant $c$ such that $\frac{v_{aj}}{v_{bj}} = c$, for $\forall j \in [n]$.
        \item There exists a set $M \subsetneq [n]$ and its complement $\overline{M}$ such that (1) $\frac{v_{aj}}{v_{bj}} = \frac{v_{ak}}{v_{bk}}$ for $\forall j, k \in M$, and $\frac{v_{aj}}{v_{bj}} = \frac{v_{ak}}{v_{bk}}$ for $\forall j, k \in \overline{M}$, and (2)
        $\frac{x_a}{x_b} = f\left(\sum_{j \in M}v_{aj}, \sum_{j \in M}v_{bj}, \sum_{j \in \overline{M}}v_{aj}, \sum_{j \in \overline{M}}v_{bj}\right)$.
    \end{itemize}
\end{theorem}

The function $f$ in Theorem \ref{finally-theorem} is complicated, so we provide its specific expression in the proof.
To reinforce the understanding of Theorem \ref{finally-theorem}, particularly the second case, we construct a class of instances and demonstrate how the utilities of both players vary across the two equilibria as their relative budgets change, as shown in Figure \ref{utility:example}.

Let the budgets of the two players be $r$ and $1$, respectively, with their relative budget ratio defined as $r = \frac{x_a}{x_b}$.
Let $v_{a1}=1$, $v_{a2}=5$, $v_{b1}=1$, and $v_{b2}=0.5$. 
Clearly, this game satisfies the second case of Theorem \ref{finally-theorem}. In Figure \ref{utility:example}, $SEu_a$ and $SEu_b$ represent the utilities of players $a$ and $b$ in the Stackelberg equilibrium, while $NEu_a$ and $NEu_b$ denote their utilities in the Nash equilibrium.
As $r$ changes, the utilities of the two players in both equilibria also change. Notably, there exists a value of $r$ such that $r = \frac{x_a}{x_b} = f(1,1,5,0.5)$, at which both players achieve the same utility in both equilibria.

Furthermore, several interesting observations emerge from these instances.
When $\frac{x_a}{x_b} > f$, both players experience higher utilities in the Stackelberg equilibrium than in the Nash equilibrium. In other words, an increase in the leader's utility does not necessarily reduce the follower's utility. Specifically, the leader's optimal commitment in the Stackelberg equilibrium allocates more resources to the higher-value battlefield (Battlefield 2) and fewer resources to the lower-value battlefield (Battlefield 1) compared to the Nash equilibrium.
Since the follower values the battlefields in the opposite way, they prefer Battlefield 1. This allows the follower to derive more utility from Battlefield 1, leading to an increase in their utility under the Stackelberg equilibrium.
Conversely, when the leader's budget is relatively small (i.e., when the ratio is less than $f$), the situation reverses. This is illustrated by the following calculation:
\begin{itemize}
    \item If $\frac{x_a}{x_b} = 0.5$, which is less than $f$, the unique Nash equilibrium is $\bm{x}_a^* = (0.025, 0.475)$, $\bm{x}_b^* = (0.340, 0.660)$, where $NEu_a = 2.161$, $NEu_b = 1.223$.
    The Stackelberg equilibrium is $\bm{x}_a = (0.136, 0.364)$, $\bm{x}_b = (0.559, 0.441)$, where $SEu_a = 2.458$, $SEu_b = 1.079$.
    \item If $\frac{x_a}{x_b} = 2$, which is greater than $f$, the unique Nash equilibrium is $\bm{x}_a^* = (0.667, 1.333)$, $\bm{x}_b^* = (0.833, 0.167)$, where $NEu_a = 4.889$, $NEu_b = 0.611$.
    The Stackelberg equilibrium is $\bm{x}_a = (0.543, 1.457)$, $\bm{x}_b = (0.847, 0.153)$, where $SEu_a = 4.915$, $SEu_b = 0.657$.
\end{itemize}

\section{Leader's Advantage} \label{leader's Advantage}
In this section, we analyze the advantages of the leader as a first mover, with these advantages being contingent upon the leader's budget relative to that of the follower.
Specifically, we consider three cases: (1) the number of battlefields $n > 2$, we find that when the leader's budget is relatively large, the advantages of the leader is marginal; (2) $n=2$, we provide an upper bound and a lower bound on the ratio of the leader's utility under the Stackelberg equilibrium to that under the Nash equilibrium; and (3) $n=2$ and the leader's budget is relatively small, we find that the advantages of the leader could be infinite.

{\bf{There are $n > 2$ battlefields.}} The following theorem gives a lower bound for the leader's utility in Nash equilibria.

\begin{theorem}\label{theorem-leader-utility-NE}
    For any game 
    $$\mathcal{G} := \langle \{a, b\}, [n], x_a, x_b, (v_{aj})_{j = 1}^{n}, (v_{bj})_{j = 1}^{n} \rangle,$$
    the leader's utility in the Nash equilibrium is greater than or equal to $\frac{x_a}{x_a + x_b} \cdot (\sum_{j = 1}^n v_{aj})$.
\end{theorem}
\begin{proof}
    Let $(\bm{x}^*_a, \bm{x}^*_b)$ denote the Nash equilibrium.
    Consider a leader's strategy $\hat{\bm{x}}_{aj} = \frac{x_a}{x_b} x^*_{bj}$, we have $u_a(\bm{x}^*_a, \bm{x}^*_b) \geq u_a(\hat{\bm{x}}_a, \bm{x}^*_b)$. Furthermore, we have
    \begin{align*}
        u_a(\hat{\bm{x}}_a, \bm{x}^*_b) = \sum_{j = 1}^n \frac{\frac{x_a}{x_b} x^*_{bj}}{\frac{x_a}{x_b} x^*_{bj} + x^*_{bj}} \cdot v_{aj} = \frac{x_a}{x_a + x_b} \sum_{j = 1}^n v_{aj}.
    \end{align*}
    It shows that the utility of the leader in the Nash equilibrium is greater than or equal to $\frac{x_a}{x_a + x_b} \cdot (\sum_{j = 1}^n v_{aj})$.
\end{proof}

By Theorem \ref{theorem-leader-utility-NE}, we have the following corollary.
\begin{corollary}\label{corollary-marginal}
    Let $u_a^{NE}$ and $u_a^{SE}$ denote the leader's utility under the Nash equilibrium and Stackelberg equilibrium, respectively.
    We have $\frac{u_a^{SE}}{u_a^{NE}} \leq \frac{x_a + x_b}{x_a}$.
\end{corollary}
\begin{proof}
   Let $(\Tilde{\bm{x}}_a, \Tilde{\bm{x}}_b)$ denote the Stackelberg equilibrium, then by Theorem \ref{theorem-leader-utility-NE}, we have $$\frac{x_a}{x_a + x_b} \sum_{j = 1}^n v_{aj} \leq u_a(\bm{x}^*_a, \bm{x}^*_b) \leq u_a(\Tilde{\bm{x}}_a, \Tilde{\bm{x}}_b) < \sum_{j = 1}^n v_{aj}.$$
   Rearranging the inequality yields $\frac{u_a(\Tilde{\bm{x}}_a, \Tilde{\bm{x}}_b)}{u_a(\bm{x}^*_a, \bm{x}^*_b)} \leq \frac{x_a + x_b}{x_a}.$
\end{proof}
This corollary shows that when $x_a$ is relatively large, $\frac{u_a(\Tilde{\bm{x}}_a, \Tilde{\bm{x}}_b)}{u_a(\bm{x}^*_a, \bm{x}^*_b)}$ approaches 1. 
Hence, when the leader's budget is relatively large, the leader can make a commitment to enhance its own benefits, but the improvement is marginal.

{\bf{There are two battlefields.}} We consider a game with two battlefields. 
Analyzing two-battlefield scenario can provide insights into the leader's advantages in games with multiple battlefields.
Let $\Tilde{\mathcal{G}} := \langle \{a, b\}, \{1, 2\}, x_a, 1, (v_{a1}, 1), (v_{b1},$ $1) \rangle$.
Without loss of generality, let $v_{a1} \leq v_{b1}$.
For the game $\Tilde{\mathcal{G}}$, \cite{Li-2022} present an approach to compute Nash equilibria. 
We can utilize this approach to compute Nash equilibria.
The following theorem gives the range of the ratio.

\begin{theorem}\label{theorem-ratio}
    In game $\Tilde{\mathcal{G}}$, let $\hat{\bm{x}}_a$ denote the leader's optimal commitment, $\hat{\bm{x}}_b$ denote the follower's best response, and $u_a^{NE}$ denote the leader's utility in the Nash equilibrium. Then, the ratio of Player $a$'s utility in the Stackelberg equilibrium to its utility in the Nash equilibrium is bounded by  
    \begin{align*}
        \frac{v_{a1} + 1}{v_{a1} + \frac{v_{b1}(x_a+1)}{v_{b1}x_a+v_{a1}}} \leq \frac{u_a(\hat{\bm{x}}_a, \hat{\bm{x}}_b)}{u_a^{NE}} \leq \frac{v_{a1} + 1}{\frac{x_av_{a1}^2}{x_av_{a1}+v_{b1}} + \frac{x_a}{x_a+1}}.
    \end{align*}
\end{theorem}

{\bf{There are two battlefields and the leader's budget is relatively small.}} By Theorem \ref{theorem-ratio}, we can derive the following corollary, which describes the leader's advantage when the leader's budget is relatively small.

\begin{corollary}\label{corollary-infinity}
    In game $\Tilde{\mathcal{G}}$, the ratio of the leader's utility when moving first to its utility in the Nash equilibrium approaches infinity.
\end{corollary}

To illustrate this corollary, we consider an extreme case $\Tilde{\mathcal{G}} := \langle \{a, b\}, \{1, 2\}, \epsilon, 1, (\epsilon, 1), (o(\epsilon), 1) \rangle$.
We observe when the leader's budget is dominated by the follower's budget, and the follower's value on the first battlefield is negligible compared to the leader's value on the first battlefield, the ratio of the leader's utility when moving first to its utility in the Nash equilibrium approaches infinity.

Although this corollary focuses on two battlefields, using the reduction method presented in Section 3, we can extend the analysis to multiple battlefields. 
Therefore, this corollary holds for scenarios with more than two battlefields as well.

\section{Conclusion and Future Work}
We explore the advantages of the first mover in the Lottery Colonel Blotto game.
To address this, we reduce the number of supports in the follower's best response strategies from $2^n-1$ to $n$.
We derive a method to calculate the leader's commitment for each possible support, allowing selection of the optimal one from these $n$ potential commitments.
In addition, we provide the necessary and sufficient conditions under which the Stackelberg equilibrium and Nash equilibrium are equivalent.
We find that the two equilibria are equivalent when the budget ratio between the leader and the follower is exactly equal to a certain functional value, which depends on the ratio of the valuations of the battlefields by the leader and the follower. 
Furthermore, we analyze the ratio of the leader's utility under the Stackelberg equilibrium to that under the Nash equilibrium. 

This work opens several directions for further research.
First, analyzing the ratio between the follower's utility under Stackelberg and Nash equilibria may yield deeper understanding of the incentives and trade-offs involved.
Second, it is of interest to study which player benefits more from committing early, and under what conditions a player prefers to act as the leader or the follower.
Finally, extending the model to settings with multiple followers introduces new layers of strategic complexity, where the interaction among followers becomes both theoretically rich and analytically challenging.

\clearpage

\section*{Acknowledgments}
We thank the fnancial support from National Natural Science Foundation of China (Grant No. 62172422 and 72192805), and Key Laboratory of Interdisciplinary Research of Computation and Economics (Shanghai University of Finance and Economics), Ministry of Education.
Jie Zhang was partially supported by a Leverhulme Trust Research Project Grant (2021 -- 2024) and the EPSRC grant (EP/W014912/1).

\bibliographystyle{named}
\bibliography{reference}

\newpage
\appendix
\renewcommand{\thelemma}{\arabic{lemma}}
\section{OMITTED PROOFS FROM SECTION 3}\label{appendix-section3}
\begin{lemma}\label{lemma-splitting}
    If the leader's commitment is $\bm{x}_a = (x_{a1}, \cdots,$ $x_{an})$, the follower's best response is $\bm{x}_b = BR(\bm{x}_a) = (x_{b1}, \cdots, x_{bn})$ in $\mathcal{G}$, and the leader's commitment is $\bm{x}_a^{(1)} = (x_{a1}, \cdots, x_{a(r-1)}, x_{a(r+1)}, \cdots, x_{an}, \frac{x_{ar}}{t}, \cdots, \frac{x_{ar}}{t})$ in $\mathcal{G}^{(1)}$, we have $\bm{x}_b^{(1)} = BR(\bm{x}_a^{(1)}) = (x_{b1}, \cdots, x_{b(r-1)}, x_{b(r+1)},$ $\cdots, x_{bn}, \frac{x_{br}}{t}, \cdots, \frac{x_{br}}{t})$.
    Furthermore, we have that $u_a(\bm{x}_a^{(1)}, $ $\bm{x}_b^{(1)}) = u_a(\bm{x}_a, \bm{x}_b)$.
\end{lemma}
\begin{proof}
    Obviously, there are $n_1 = n - 1 + t$ battlefields in $\mathcal{G}^{(1)}$.
    By Lemma \ref{sec2-lemma-br}, there exists a set $K(\bm{x}_a)$ in $\mathcal{G}$.
    Therefore, there are two cases regarding $r$: (1) If $r \notin K(\bm{x}_a)$, we have $x_{bn+1}^{(1)} = x_{bn+2}^{(1)} = \cdots = x_{bn+t}^{(1)} = 0$, and $K(\bm{x}_a^{(1)}) = K(\bm{x}_a)$.
    (2) If $r \in K(\bm{x}_a)$, we have $x_{bn+1}^{(1)} = x_{bn+2}^{(1)} = \cdots = x_{bn+t}^{(1)} = \frac{x_{br}}{t}$, and $K(\bm{x}_a^{(1)}) = K(\bm{x}_a) \backslash \{r\} \cup \{n + 1, \cdots, n + t \}$.
    We substitute $K(\bm{x}_a^{(1)})$, $\bm{x}_b^{(1)}$ and $\bm{x}_a^{(1)}$ into Equation \eqref{sec2-lem-equ-br}, and it is easy to verify that Equation \eqref{sec2-lem-equ-br} is established.
    
    The aggregate utility from battlefields $\{1, 2, \cdots, r - 1, r + 1, \cdots, n\}$ is equal for players in both $\mathcal{G}$ and $\mathcal{G}^{(1)}$.
    Consider the last $t$ battlefields in $\mathcal{G}^{(1)}$, we have that 
    $$\sum_{j = n + 1}^{n + t} u_{bj}(\bm{x}_a^{(1)}, \bm{x}_b^{(1)}) = \sum_{j = n + 1}^{n + t} \frac{\frac{x_{br}}{t}}{\frac{x_{ar}}{t} + \frac{x_{br}}{t}} \frac{v_{br}}{t} = u_{br}(\bm{x}_a, \bm{x}_b).$$
    Similarly, we can also obtain $\sum_{j = n + 1}^{n + t} u_{aj}(\bm{x}_a^{(1)}, \bm{x}_b^{(1)}) = u_{ar}(\bm{x}_a, \bm{x}_b)$.
    Therefore, $u_a(\bm{x}_a^{(1)}, \bm{x}_b^{(1)}) = u_a(\bm{x}_a, \bm{x}_b)$.
\end{proof}

\begin{lemma}\label{lemma-highervalue-resource}
    In $\mathcal{G}^{(2)}$, the leader allocates at least the same amount of resources to the battlefields with higher leader's valuations in the optimal commitment.
\end{lemma}
\begin{proof}
    We prove it by contradiction.
    Let $\bm{x}_{a}^{(2)}$ and $\bm{x}_{b}^{(2)}$ denote the leader's commitment and the follower's best response respectively. 
    If there are two distinct battlefields $j$ and $j'$ satisfying $x_{aj}^{(2)} > x_{aj'}^{(2)}$ and $v_{aj}^{(2)} < v_{aj'}^{(2)}$.
    Consider another commitment $\bm{\hat{x}}_{a}^{(2)}$ where $\hat{x}_{aj}^{(2)} = x_{aj'}^{(2)}$, $\hat{x}_{aj'}^{(2)} = x_{aj}^{(2)}$ and $\hat{x}_{ah}^{(2)} = x_{ah}^{(2)}$ for $\forall h \neq j, j'$.
    By Lemma \ref{sec2-lemma-br}, it is easy to verify that $\bm{\hat{x}}_{b}^{(2)} = BR(\bm{\hat{x}}_{a}^{(2)})$ where $\hat{x}_{bj}^{(2)} = x_{bj'}^{(2)}$, $\hat{x}_{bj'}^{(2)} = x_{bj}^{(2)}$ and $\hat{x}_{bh}^{(2)} = x_{bh}^{(2)}$ for $\forall h \neq j, j'$.
    In $\hat{\bm{x}}_{a}^{(2)}$, the leader and the follower allocate resources to battlefields other than $j$ and $j'$ in the same manner as in $\bm{x}_{a}^{(2)}$; hence, our analysis needs to focus solely on the aggregate utility at battlefields $j$ and $j'$.
    In both $\bm{x}_a^{(2)}$ and $\hat{\bm{x}}_a^{(2)}$, the difference in the aggregate utility from battlefield $j$ and battlefield $j'$ is as follows
    \begin{align}\label{equation-highvalue-resource}
        &\frac{x_{aj'}^{(2)} v_{aj}^{(2)}}{x_{aj'}^{(2)} + x_{bj'}^{(2)}} + \frac{x_{aj}^{(2)} v_{aj'}^{(2)}}{x_{aj}^{(2)} + x_{bj}^{(2)}} - \frac{x_{aj}^{(2)} v_{aj}^{(2)}}{x_{aj}^{(2)} + x_{bj}^{(2)}} - \frac{x_{aj'}^{(2)} v_{aj'}^{(2)}}{x_{aj'}^{(2)} + x_{bj'}^{(2)}}  \nonumber \\
        = &(v_{aj'}^{(2)} - v_{aj}^{(2)}) \left(\frac{x_{aj}^{(2)}}{x_{aj}^{(2)} + x_{bj}^{(2)}} - \frac{x_{aj'}^{(2)}}{x_{aj'}^{(2)} + x_{bj'}^{(2)}} \right).
    \end{align}
    By Equation \eqref{sec2-lem-equ-br}, we obtain $x_{bj}^{(2)}$ and $x_{bj'}^{(2)}$ respectively,
    
    \begin{equation*}
        \begin{cases}
            x_{bj}^{(2)} = \frac{(x_{aj}^{(2)})^\frac{1}{2}}{\sum_{l \in K(\bm{x}_a^{(2)})} (x_{al}^{(2)})^\frac{1}{2}}  (x_b + \sum_{l \in K(\bm{x}_a^{(2)})} x_{al}^{(2)}) - x_{aj}^{(2)};  \\
            x_{bj'}^{(2)} = \frac{(x_{aj'}^{(2)})^\frac{1}{2}}{\sum_{l \in K(\bm{x}_a^{(2)})} (x_{al}^{(2)})^\frac{1}{2}}  (x_b + \sum_{l \in K(\bm{x}_a^{(2)})} x_{al}^{(2)}) - x_{aj'}^{(2)}.
        \end{cases}
    \end{equation*}
    We substitute $x_{bj}^{(2)}$ and $x_{bj'}^{(2)}$ into Equation \eqref{equation-highvalue-resource}, and we find that the value of Equation \eqref{equation-highvalue-resource} is positive.
    Therefore, $u_a(\bm{x}_{a}^{(2)}, \bm{x}_{b}^{(2)}) < u_a(\bm{\hat{x}}_{a}^{(2)}, \bm{\hat{x}}_{b}^{(2)})$.
    This immediately implies that for the leader, allocating at least the same amount of resources to the battlefields with higher value within his optimal commitment.
\end{proof}

\begin{lemma}\label{lemma-equalvalue-resource}
    In $\mathcal{G}^{(2)}$, the leader allocates equal resources to battlefields with equal leader's valuations in the optimal commitment.
\end{lemma}
\begin{proof}
    We prove it by contradiction. 
    Let $j$ denote the battlefield with the smallest index such that $x_{aj}^{(2)} < x_{aj+1}^{(2)}$ and $v_{aj}^{(2)} = v_{aj+1}^{(2)}$.
    Consider another commitment $\bm{\hat{x}}_{a}^{(2)}$ where $\hat{x}_{aj}^{(2)} = \hat{x}_{aj+1}^{(2)} = \frac{1}{2} (x_{aj}^{(2)} + x_{aj+1}^{(2)})$ and $\hat{x}_{ah}^{(2)} = x_{ah}^{(2)}$ for $\forall h \neq j, j+1$.
    Accordingly, let $\bm{\hat{x}}_{b}^{(2)} = BR(\bm{\hat{x}}_{a}^{(2)})$.
    Note that if battlefield $j$ is not in $K(\bm{x}_a^{(2)})$, then battlefield $j+1$ is also not in $K(\bm{x}_a^{(2)})$.
    Therefore, there are three cases for $K(\bm{x}_a^{(2)})$, namely (1) $j \in K(\bm{x}_a^{(2)})$ and $j + 1 \in K(\bm{x}_a^{(2)})$, (2) $j \notin K(\bm{x}_a^{(2)})$ and $j + 1 \notin K(\bm{x}_a^{(2)})$, and (3) $j \in K(\bm{x}_a^{(2)})$ and $j + 1 \notin K(\bm{x}_a^{(2)})$.
    For each case, we discuss the changes between $K(\bm{x}_a^{(2)})$ and $K(\bm{\hat{x}}_a^{(2)})$, and compare $u_a(\bm{x}_{a}^{(2)}, \bm{x}_{b}^{(2)})$ and $u_a(\bm{\hat{x}}_{a}^{(2)}, \bm{\hat{x}}_{b}^{(2)})$. 
    
    For the case (1): according to Lemma \ref{sec2-lemma-br}, we have
    \begin{equation}\label{equation-sec2-lemma-br-k}
        k^* = \max\{k \in [n_2]: x_b + \sum_{l = 1}^{j} x_{al}^{(2)} > (x_{aj}^{(2)})^{\frac{1}{2}} \sum_{l = 1}^{j} (x_{al}^{(2)})^{\frac{1}{2}}, \, \forall j \leq k  \}.
    \end{equation}
    Due to $j \in K(\bm{x}_a^{(2)})$ and $j + 1 \in K(\bm{x}_a^{(2)})$, we immediately obtain the following inequality by Equation \eqref{equation-sec2-lemma-br-k},
    \begin{equation*}
        \left\{
            \begin{aligned}
                x_b + \sum_{l = 1}^{j-1} x_{al}^{(2)} + x_{aj}^{(2)} > &(x_{aj}^{(2)})^{\frac{1}{2}} \left( \sum_{l = 1}^{j-1} (x_{al}^{(2)})^{\frac{1}{2}} + (x_{aj}^{(2)})^{\frac{1}{2}} \right); \\
                x_b + \sum_{l = 1}^{j-1} x_{al}^{(2)} + x_{aj}^{(2)} + x_{aj+1}^{(2)} > &(x_{aj+1}^{(2)})^{\frac{1}{2}} \left( \sum_{l = 1}^{j-1} (x_{al}^{(2)})^{\frac{1}{2}} + (x_{aj}^{(2)})^{\frac{1}{2}} + (x_{aj+1}^{(2)})^{\frac{1}{2}} \right).
            \end{aligned}
        \right.
    \end{equation*}
    Due to $x_{aj}^{(2)} < x_{aj+1}^{(2)}$, we can derive the following inequality
    \begin{equation*}
        \left\{
            \begin{aligned}
                x_b + \sum_{l = 1}^{j-1} x_{al}^{(2)} + \frac{x_{aj}^{(2)} + x_{aj+1}^{(2)}}{2} &> \left(\frac{x_{aj}^{(2)} + x_{aj+1}^{(2)}}{2}\right)^{\frac{1}{2}} \left( \sum_{l = 1}^{j-1} \left(x_{al}^{(2)}\right)^{\frac{1}{2}} + \left(\frac{x_{aj}^{(2)} + x_{aj+1}^{(2)}}{2}\right)^{\frac{1}{2}} \right); \\
                x_b + \sum_{l = 1}^{j-1} x_{al}^{(2)} + x_{aj}^{(2)} + x_{aj+1}^{(2)} &> \left(\frac{x_{aj}^{(2)} + x_{aj+1}^{(2)}}{2}\right)^{\frac{1}{2}} \left( \sum_{l = 1}^{j-1} \left(x_{al}^{(2)}\right)^{\frac{1}{2}} + 2\left(\frac{x_{aj}^{(2)} + x_{aj+1}^{(2)}}{2}\right)^{\frac{1}{2}} \right).
            \end{aligned}
        \right.
    \end{equation*}
    The above inequality implies that $j \in K(\bm{\hat{x}}_a^{(2)})$ and $j+1 \in K(\bm{\hat{x}}_a^{(2)})$.
    
    We analyze the changes in the leader's utility from $\bm{x}_a^{(2)}$ to $\bm{\hat{x}_a}^{(2)}$.
    Consider $\bm{x}_a^{(2)\delta} = (x_{a1}^{(2)}, \cdots, x_{aj}^{(2)} + \delta, x_{aj+1}^{(2)} - \delta, \cdots, x_{an_2}^{(2)})$.
    If $\delta = \frac{x_{aj+1}^{(2)} - x_{aj}^{(2)}}{2}$, we have $\bm{\hat{x}}_a^{(2)} = \bm{x}_a^{(2)\delta}$. 
    Let $K(\bm{x}_a^{(2)\delta})$ denote the battlefield set corresponding to $\bm{x}_a^{(2)\delta}$ and $\bm{x}_b^{(2)\delta} = BR(\bm{x}_a^{(2)\delta})$. 
    Due to Equation \eqref{equation-sec2-lemma-br-k}, we can know that the set $K(\bm{x}_a^{(2)\delta})$ does not expand as $\delta$ increases. 
    Suppose that (i) $\delta_1 \leq \delta < \delta_2$, we have $K(\bm{x}_a^{(2)\delta}) = K(\bm{x}_a^{(2)\delta_1})$, and (ii) when $r = |K(\bm{x}_a^{(2)\delta_2})|$, we have $x_b + \sum_{l = 1}^{r} x_{al}^{(2)\delta_2} = (x_{ar}^{(2)\delta_2})^{\frac{1}{2}} \sum_{l = 1}^{r}(x_{al}^{(2)\delta_2})^{\frac{1}{2}}$, which means that the battlefield $\{r\} = K(\bm{x}_a^{(2)\delta}) \backslash K(\bm{x}_a^{(2)\delta_2})$ where $\delta \in [\delta_1, \delta_2)$. 
    Note that both $\delta = \delta_2$ and $\delta_1 \leq \delta < \delta_2$ work for Equation \eqref{equation-sec2-lemma-br-k} because of continuity of $\bm{x}_a^{(2)}$, which can be viewed as battlefield $r$ is in the set $K(\bm{x}_a^{(2)\delta_2})$ but $\bm{x}_{br}^{(2)\delta_2}$ is exactly zero. 
    According to Equation \eqref{sec2-equ-ui}, we have the leader's aggregate utility $\sum_{j \in K(\bm{x}_a^{(2)\delta})} \frac{x_{aj}^{(2)\delta} v_{aj}^{(2)}}{x_{aj}^{(2)\delta} + x_{bj}^{(2)\delta}}$ in set $K(\bm{x}_a^{(2)\delta})$, and we substitute the follower's best response into the leader's aggregate utility in the set $K(\bm{x}_a^{(2)\delta})$ by Equation \eqref{sec2-lem-equ-br}, obtaining that
    \begin{align*}
        \tilde{u}_a(\bm{x}_a^{(2)\delta}) &= \left( \sum_{j \in K(\bm{x}_a^{(2)\delta})} (\frac{x_{aj}^{(2)\delta}}{v_{bj}^{(2)}})^{\frac{1}{2}} \cdot v_{aj}^{(2)} \right) \left( \frac{\sum_{j \in K(\bm{x}_a^{(2)\delta})} (x_{aj}^{(2)\delta} \cdot v_{bj}^{(2)})^\frac{1}{2}}{\sum_{j \in K(\bm{x}_a^{(2)\delta})} x_{aj}^{(2)\delta} + x_b} \right)  \\
        &= \left( \sum_{j \in K(\bm{x}_a^{(2)\delta})} (x_{aj}^{(2)\delta})^{\frac{1}{2}} \cdot v_{aj}^{(2)} \right)  \cdot \left( \frac{\sum_{j \in K(\bm{x}_a^{(2)\delta})} (x_{aj}^{(2)\delta})^\frac{1}{2}}{\sum_{j \in K(\bm{x}_a^{(2)\delta})} x_{aj}^{(2)\delta} + x_b} \right)
    \end{align*}
    We take the partial derivative of $\tilde{u}_a(\bm{x}_a^{(2)\delta})$ with regard to $\bm{x}_{aj}^{(2)\delta}$, that is
    \begin{align*}
        \frac{\partial \tilde{u}_a(\bm{x}_a^{(2)\delta})}{\partial x_{aj}^{(2)\delta}} = & \frac{1}{x_b + \sum_{j \in K(\bm{x}_a^{(2)\delta})} x_{aj}^{(2)\delta}} \bigg[ \frac{1}{2} \left(\sum_{j \in K(\bm{x}_a^{(2)\delta})}(x_{aj}^{(2)\delta})^\frac{1}{2}\right) v_{aj}^{(2)} \times (x_{aj}^{(2)\delta})^{-\frac{1}{2}} \\
        & + (\sum_{j \in K(\bm{x}_a^{(2)\delta})} v_{aj}^{(2)} \sqrt{x_{aj}^{(2)\delta}}) \times \left( \frac{\frac{1}{2} (x_{aj}^{(2)\delta})^{-\frac{1}{2}} (x_b + \sum_{j \in K(\bm{x}_a^{(2)\delta})} x_{aj}^{(2)\delta}) - \sum_{j \in K(\bm{x}_a^{(2)\delta})} (x_{aj}^{(2)\delta})^\frac{1}{2}}{x_b + \sum_{j \in K(\bm{x}_a^{(2)\delta})} x_{aj}^{(2)\delta}} \right) \bigg]
    \end{align*}
    We can find that when $\delta \in [\delta_1, \delta_2]$, $\frac{\partial \tilde{u}_a(\bm{x}_a^{(2)\delta})}{\partial \bm{x}_{aj}^{(2)\delta}} > \frac{\partial \tilde{u}_a(\bm{x}_a^{(2)\delta})}{\partial \bm{x}_{aj+1}^{(2)\delta}}$. Therefore, the utility of the leader is improved from $\bm{x}_a^{(2)\delta_1}$ to $\bm{x}_a^{(2)\delta_2}$. We keep increasing $\delta$ until it equals $\frac{x_{aj+1}^{(2)} - x_{aj}^{(2)}}{2}$, and the leader's utility is increasing. 

    For the case (2): according to Equation \eqref{equation-sec2-lemma-br-k}, we find that the set $K(\bm{x}_a^{(2)}) = K(\bm{\hat{x}}_a^{(2)})$. Given that $x_{aj'}^{(2)} = \hat{x}_{aj'}^{(2)}$ for $\forall j' \in K(\bm{x}_a^{(2)})$, it follows that the follower's best response remains unchanged, consequently, the leader's utility remains unchanged.

    For the case (3): since the set $K(\bm{x}_a^{(2)})$ will not expand, the set $K(\bm{x}_a^{(2)\delta})$ will undergo following change process: as $\delta$ increases from 0 up to a certain threshold $\delta^*$, we have $K(\bm{x}_a^{(2)}) = K(\bm{x}_a^{(2)\delta})$. However, if $\delta$ continues to increase beyond $\delta^*$, the set becomes $K(\bm{x}_a^{(2)}) \backslash \{|K(\bm{x}_a^{(2)})|\} = K(\bm{x}_a^{(2)\delta})$. During the first stage of changes, the situation is analogous to case (2), and thus the utility remains unchanged. In the second stage of changes, the situation mirror case (1), leading to an increase in utility.
    This immediately implies that for the leader, allocating equal resources to the battlefields with equal valuation within his optimal commitment.
\end{proof}

\begin{lemma}\label{lemma-merging}
    In $\mathcal{G}^{(2)}$, if $v_{ar}^{(2)} = v_{a(r+1)}^{(2)} = \cdots = v_{a(r+t-1)}^{(2)}$ and the leader's commitment is $\bm{x}_a^{(2)}$ where $x_{aj}^{(2)} = x_{ah}^{(2)}$ for $\forall j, h \in \{r, r+1, \cdots, r+t-1 \}$.
    Then in $\mathcal{G}^{(3)}$, for the commitment $\bm{x}_{a}^{(3)}$ where $x_{ar}^{(3)} = t x_{ar}^{(2)}$, $x_{aj}^{(3)} = x_{aj}^{(2)}$ for $\forall j \in \{1, 2, \cdots, r-1 \}$ and $x_{ah-t+1}^{(3)} = x_{ah}^{(2)}$ for $\forall h \in \{r+t, r+t+1, \cdots, n_2 \}$, we have $\bm{x}_{b}^{(3)} = BR(\bm{x}_{a}^{(3)})$ where $x_{br}^{(3)} = t x_{br}^{(2)}$ and $x_{bj}^{(3)} = x_{bj}^{(2)}$ for $\forall j \in \{1, 2, \cdots, r-1 \}$ and $x_{bh-t+1}^{(3)} = x_{bh}^{(2)}$ for $\forall h \in \{r+t, r+t+1, \cdots, n_2 \}$.
    Furthermore, $u_a(\bm{x}_{a}^{(2)}, \bm{x}_{b}^{(2)}) = u_a(\bm{x}_{a}^{(3)}, \bm{x}_{b}^{(3)})$.
\end{lemma}
\begin{proof}
    Similar to the proof method outlined in Lemma \ref{lemma-splitting}, there exists a set $K(\bm{x}_a^{(2)})$ according to Lemma \ref{sec2-lemma-br}. There are two cases: (1) if $x_{br}^{(2)} = 0$, it follows that $x_{br}^{(3)} = 0$, and consequently, $K(\bm{x}_a^{(2)}) = K(\bm{x}_a^{(3)})$, and (2) if $x_{br}^{(2)} > 0$, it follows that $x_{br}^{(3)} > 0$, and consequently, $K(\bm{x}_a^{(3)}) = K(\bm{x}_a^{(2)}) \backslash \{ r + 1, \cdots, r + t - 1\}$. Substituting $K(\bm{x}_a^{(3)})$, $\bm{x}_b^{(3)}$ and $\bm{x}_a^{(3)}$ into Equation \eqref{sec2-lem-equ-br}, it is straightforward to confirm that Equation \eqref{sec2-lem-equ-br} holds. By Equation \eqref{sec2-equ-uij}, we find that $\sum_{j = r}^{r+t-1)} u_{aj}(x_{aj}^{(2)}, x_{bj}^{(2)}) = u_{ar}(x_{ar}^{(3)}, x_{br}^{(3)})$, which leads to the conclusion $u_a(\bm{x}_a^{(2)}, \bm{x}_b^{(2)}) = u_a(\bm{x}_a^{(3)}, \bm{x}_b^{(3)})$.
\end{proof}

{\noindent \bf Proof of Theorem \ref{theorem-characterization}}
\begin{proof}
    Let $\mathcal{G} := \langle \{a, b\}, [n], x_a, x_b, (v_{aj})_{j = 1}^{n}, (v_{bj})_{j = 1}^{n} \rangle$ denote this game.
    Then we split every battlefield, and the resulting auxiliary game is denoted as 
    $$\mathcal{G}^{(2)} := \langle \{a, b\}, [n_2], x_a, x_b, (v_{aj}^{(2)})_{j = 1}^{n_2}, (v_{bj}^{(2)})_{j = 1}^{n_2} \rangle.$$
    In $\mathcal{G}^{(2)}$, the follower's valuation for each battlefield is the same, and for the leader, for a certain battlefield $r$ in $\mathcal{G}$ and the sub-battlefields in $\mathcal{G}^{(2)}$ corresponding to battlefield $r$, the leader's valuation of these sub-battlefields is the same.
    Let $\bm{x}^*_a$ denote the leader's optimal commitment in $\mathcal{G}$.
    Let $\bm{x}_a^{*(2)}$ represent the leader's commitment in $\mathcal{G}^{(2)}$ that corresponds to $\bm{x}^*_a$.
    Let $\bm{x}_a^{(2)*}$ denote the optimal commitment in $\mathcal{G}^{(2)}$.
    According to Lemma \ref{lemma-splitting}, we obtain
    $$
    u_a(x_a^*, BR(x_a^*)) = u_a(\bm{x}_a^{*(2)}, BR(\bm{x}_a^{*(2)})) \leq u_a(\bm{x}_a^{(2)*}, BR(\bm{x}_a^{(2)*}))
    $$
    Due to $\frac{v_{a1}}{v_{b1}} \leq \frac{v_{a2}}{v_{b2}} \leq \cdots \leq \frac{v_{an}}{v_{bn}}$, it follows that in $\mathcal{G}^{(2)}$, we have $v_{a1}^{(2)} \leq v_{a2}^{(2)} \leq \cdots \leq v_{an_2}^{(2)}$.
    According to Lemma \ref{lemma-highervalue-resource}, it can be deduced that $ x_{a1}^{(2)*} \leq x_{a2}^{(2)*} \leq \cdots \leq x_{an_2}^{(2)*}$.
    According to Lemma \ref{lemma-equalvalue-resource}, we can infer that if $v_{aj}^{(2)} = v_{aj+1}^{(2)}$, then $x_{aj}^{(2)*} = x_{aj+1}^{(2)*}$, for all $j \in [n_2 - 1]$.
    For any battlefield \( r \) in \( \mathcal{G} \), we identify the corresponding battlefields in \( \mathcal{G}^{(2)} \) and merge these battlefields into a single one.
    By merging battlefields, we can eventually restore \( \mathcal{G}^{(2)} \) to the original game \( \mathcal{G} \).
    Let $\Tilde{\bm{x}}_a^*$ denote the commitment in $\mathcal{G}$ that corresponds to $\bm{x}_a^{(2)*}$. Note that $\Tilde{\bm{x}}_a^*$ is a commitment in $\mathcal{G}$.
    According to Lemma \ref{lemma-merging}, we can deduce that
    $$
    u_a(\bm{x}_a^{(2)*}, BR(\bm{x}_a^{(2)*})) = u_a(\Tilde{\bm{x}}_a^*, BR(\Tilde{\bm{x}}_a^*))
    $$
    Due to $\bm{x}^*_a$ is a optimal commitment in $\mathcal{G}$, we have
    $$
    u_a(\Tilde{\bm{x}}_a^*, BR(\Tilde{\bm{x}}_a^*)) \leq u_a(\bm{x}^*_a, BR(\bm{x}^*_a))
    $$
    Finally, we have following relational expression
    \begin{align*}
        u_a(x_a^*, BR(x_a^*)) &= u_a(\bm{x}_a^{*(2)}, BR(\bm{x}_a^{*(2)})) \\
        &\leq u_a(\bm{x}_a^{(2)*}, BR(\bm{x}_a^{(2)*})) \\
        &= u_a(\Tilde{\bm{x}}_a^*, BR(\Tilde{\bm{x}}_a^*)) \\
        &\leq u_a(\bm{x}^*_a, BR(\bm{x}^*_a))
    \end{align*}
    This implies that
    \begin{align*}
        u_a(x_a^*, BR(x_a^*)) &= u_a(\bm{x}_a^{*(2)}, BR(\bm{x}_a^{*(2)})) \\
        &= u_a(\bm{x}_a^{(2)*}, BR(\bm{x}_a^{(2)*})) \\
        &= u_a(\Tilde{\bm{x}}_a^*, BR(\Tilde{\bm{x}}_a^*)) \\
        &= u_a(\bm{x}^*_a, BR(\bm{x}^*_a))
    \end{align*}
    Thus, the optimal commitment \( x_a^* \) in \( \mathcal{G} \), and its corresponding \( \bm{x}_a^{*(2)} \), is the optimal commitment in \( \mathcal{G}^{(2)} \).
    For battlefield \( j \) in \( \mathcal{G} \) where \( j \in [n] \), the ratio \( \frac{x_{aj}^*}{v_{bj}} \) remains unchanged before and after the battlefield is split and merged.
    According to Lemma \ref{lemma-highervalue-resource}, \( x_{a1}^{(2)*} \leq x_{a2}^{(2)*} \leq \cdots \leq x_{an_2}^{(2)*} \), we can deduce that for all \( j \in [n] \):
    $$
    \frac{x_{a1}^*}{v_{b1}} \leq \frac{x_{a2}^*}{v_{b2}} \leq \cdots \leq \frac{x_{an}^*}{v_{bn}}
    $$
    Hence,
    $$
    \frac{v_{b1}}{x_{a1}^*} \geq \frac{v_{b2}}{x_{a2}^*} \geq \cdots \geq \frac{v_{bn}}{x_{an}^*}
    $$
    According to Lemma \ref{sec2-lemma-br}, we can deduce that there are \( n \) possible values for \( K \), which are $ \{1\}, \{1,2\}, \ldots,$  \\  $\{1,2,3, \ldots, n\}$.
    In the game \( \mathcal{G} \), for \( j, h \in [n] \), if \( \frac{v_{aj}}{v_{bj}} = \frac{v_{ah}}{v_{bh}} \), then the leader allocates the same amount of resources to the corresponding sub-battlegrounds of battlefields \( j \) and \( h \), respectively.
    Therefore, we have \( \frac{x_{aj}^*}{x_{ah}^*} = \frac{v_{aj}}{v_{ah}} \).
\end{proof}

\section{OMITTED PROOFS FROM SECTION 4}\label{appendix-section4}
{\bf Proof of Lemma \ref{sec4-lemma-leader-threshold}}
\begin{proof}
    By Lemma \ref{sec2-lemma-br}, we have $\frac{v_{bj}}{x_{aj}} \leq \frac{\left( \sum_{l \in K} (x_{al} v_{bl})^{\frac{1}{2}} \right)^2}{\left( x_b + \sum_{l \in K} x_{al} \right)^2}$ for $j \in \overline{K}$ and rearrange it, we have $x_{aj} \geq \frac{v_{bj} \left( x_b + \sum_{l \in K} x_{al} \right)^2}{\left( \sum_{l \in K} (x_{al} v_{bl})^{\frac{1}{2}} \right)^2}$.
    Below we show that for a given $K$, $x_{aj}$ cannot be greater than $\frac{v_{bj} \left( x_b + \sum_{l \in K} x_{al} \right)^2}{\left( \sum_{l \in K} (x_{al} v_{bl})^{\frac{1}{2}} \right)^2}$ within the leader's optimal commitment for $j \in \overline{K}$.

    We choose a battlefield $j \in \overline{K}$ such that $j \in \mathop{\arg\max}\limits_{l \in \overline{K}} x_{al}$.
    Then we transfer an amount $\varepsilon$ of resources from battlefield $j$ to a certain battlefield in $K$, while transferring an amount $\varepsilon$ of resources to other battlefields in $\overline{K}$.
    The purpose of transferring resources to other battlefields in $\overline{K}$ is to ensure that $K$ remains unchanged when transferring resources to a certain battlefield in $K$.
    We now prove that after transferring resources, the leader's utility increases.
    
    If the leader allocates more resources to battlefield $j \in \overline{K}$ than $\frac{v_{bj} \left( x_b + \sum_{l \in K} x_{al} \right)^2}{\left( \sum_{l \in K} (x_{al} v_{bl})^{\frac{1}{2}} \right)^2}$, for this battlefield $j$, we transfer an amount $\varepsilon$ of resources to some battlefield in $K$. If $\varepsilon$ is sufficiently small, $K$ remains unchanged. We demonstrate that the leader's utility improves after reallocating resources. First, when the leader chooses strategy $\bm{x}_a$ and the follower chooses the best response $\bm{x}_b$, the leader's utility on battlefield $j \in K$ is $u_{aj}(\bm{x}_a, \bm{x}_b) = \frac{x_{aj} v_{aj}}{x_{aj} + x_{bj}}$. According to Lemma \ref{sec2-lemma-br}, substituting the follower's best response can yield a utility function solely in terms of $\bm{x}_a$,
    $$
    u_{aj}(\bm{x}_a) = \left(\frac{x_{aj}}{v_{bj}}\right)^{\frac{1}{2}} \cdot v_{aj} \cdot \left( \frac{\sum_{h \in K} (x_{ah} \cdot v_{bh})^\frac{1}{2}}{x_b + \sum_{h \in K} x_{ah}} \right).
    $$
    The aggregate utility of the leader over $K$ is
    $$
    u_{aK}(\bm{x}_a) = \sum_{j \in K} u_{aj}(\bm{x}_a) = \frac{\sum_{h \in K}(x_{ah} v_{bh})^\frac{1}{2}}{x_b + \sum_{h \in K} x_{ah}} \sum_{j \in K} \left(\frac{x_{aj}}{v_{bj}}\right)^\frac{1}{2} v_{aj}.
    $$
    If the leader reallocates an infinitesimal amount $\varepsilon$ of resources to battlefield $j \in K$, the partial derivative is given by
    $$
    \frac{\partial u_{aK}}{\partial x_{aj}} = \frac{\frac{1}{2} (x_{aj})^{-\frac{1}{2}} (v_{bj})^\frac{1}{2}}{\sum_{h \in K} (x_{ah} v_{bh})^\frac{1}{2}} - \frac{1}{x_b + \sum_{h \in K} x_{ah}} + \frac{\frac{1}{2} \frac{(x_{aj})^{-\frac{1}{2}}}{(v_{bj})^\frac{1}{2}} v_{aj}}{\sum_{h \in K} \left(\frac{x_{ah}}{v_{bh}}\right)^\frac{1}{2} v_{ah}}.
    $$
    Thus, we obtain the following expression
    $$
    x_{aj} \frac{\partial u_{aK}}{\partial x_{aj}} = \frac{\frac{1}{2} (x_{aj})^\frac{1}{2} (v_{bj})^\frac{1}{2}}{\sum_{h \in K} (x_{ah} v_{bh})^\frac{1}{2}} - \frac{x_{aj}}{x_b + \sum_{h \in K} x_{ah}} + \frac{\frac{1}{2} \frac{(x_{aj})^\frac{1}{2}}{(v_{bj})^\frac{1}{2}} v_{aj}}{\sum_{h \in K} \left(\frac{x_{ah}}{v_{bh}}\right)^\frac{1}{2} v_{ah}}.
    $$
    Further, we derive
    $$
    \sum_{j \in K} x_{aj} \frac{\partial u_{aK}}{\partial x_{aj}} = 1 - \frac{\sum_{h \in K} x_{ah}}{x_b + \sum_{h \in K} x_{ah}} > 0.
    $$
    
    This implies that within the first $K$ battlefields, there must exist a battlefield $r \in K$ such that $\frac{\partial u_{aK}}{\partial x_{ar}} > 0$. In other words, if there exists a battlefield $j \in \overline{K}$ where the leader allocates more resources than $\frac{v_{bj} \left( x_b + \sum_{l \in K} x_{al} \right)^2}{\left( \sum_{l \in K} (x_{al} v_{bl})^{\frac{1}{2}} \right)^2}$, the leader can transfer $\varepsilon$ resources from this battlefield to battlefield $r$, thereby strictly increasing the leader's utility. Consequently, for all $j \in \overline{K}$, the resources allocated by the leader are precisely equal to $\frac{v_{bj} \left( x_b + \sum_{l \in K} x_{al} \right)^2}{\left( \sum_{l \in K} (x_{al} v_{bl})^{\frac{1}{2}} \right)^2}$.
\end{proof}

{\noindent \bf Proof of Lemma \ref{sec4-lemma-parameters}}
\begin{proof}
     Let $|K(\bm{x}_a)| = k$, $\Delta \bm{x} = (\Delta x_{j})_{j = 1}^k$. Consider another commitment $\bm{x'}_a = ((x_{aj} - \Delta x_j)_{j = 1}^k, x_{a(k+1)}, \cdots, x_{an})$ and $\bm{x'}_b = BR(\bm{x'}_a)$, and the corresponding battlefield set is $K(\bm{x'}_a)$.
     Note that $\bm{x}_a$ is the optimal commitment, and let $\bm{x}_b = BR(\bm{x}_a)$, therefore, we have that $u_a(\bm{x'}_a, \bm{x'}_b) \leq u_a(\bm{x}_a, \bm{x}_b)$.
     By Lemma \ref{sec2-lemma-br}, we can know that the follower's best response is a continuous function of $\bm{x}_a$.
     Note that $x_{bj}>0$ for $j\in K(\bm{x}_a)$.
     Therefore, as long as $\Delta \bm{x}$ is small enough, we still have that $x_{bj} > 0$, $\forall j \in K(\bm{x}_a)$.
     There are two cases: (1) $K(\bm{x}_a) = K(\bm{x'}_a)$, and (2) $K(\bm{x}_a) \subsetneq K(\bm{x'}_a)$.
     
     If $K(\bm{x'}_a)$ is greater than $K(\bm{x}_a)$, the follower must choose some battlefields in $\overline{K(\bm{x}_a)}$. Note that the action profile $(\bm{x}_a, \bm{x}_b)$, the battlefields in the set $\overline{K(\bm{x}_a)}$ have the same marginal utility for the follower, then the follower will simultaneously invest resources in these battlefields. Consequently, $K(\bm{x'}_a) = [n]$ when $K(\bm{x'}_a)$ is greater than $K(\bm{x}_a)$. Now we only discuss the case (1) and the case (2) is similar to the case (1), thus the case (2) can be omitted.
    
    First, we note that $\bm{x}_a$ is the optimal commitment. When the leader makes the commitment $\bm{x'}_a$, the leader's utility do not increase. Therefore, we have the following inequality
    \begin{equation*}
        \left( \frac{\partial \hat{u}_a(\bm{x}_a)}{\partial x_{a1}}, \frac{\partial \hat{u}_a(\bm{x}_a)}{\partial x_{a2}}, \cdots, \frac{\partial \hat{u}_a(\bm{x}_a)}{\partial x_{ak}} \right) \cdot \Delta \bm{x} \leq 0.
    \end{equation*}
    Due to the set $K(\bm{x}_a)$ does not changed and the battlefields in the set $\overline{K(\bm{x}_a)}$ does not enter the set $K(\bm{x}_a)$, according to Lemma \ref{sec2-lemma-br}, we have that $\sum_{h \in K(\bm{x}_a)}(x_{ah}' v_{bh})^\frac{1}{2} \leq \sum_{h \in K(\bm{x}_a)}(x_{ah} v_{bh})^\frac{1}{2}$. Let $L(\bm{x}_a) = \sum_{h \in K(\bm{x}_a)}(x_{ah} v_{bh})^\frac{1}{2}$. For any $\Delta \bm{x}$ satisfying $(1, 1, \cdots, 1) \cdot \Delta \bm{x} = 0$, which means that the amount of resources allocated by the leader remains unchanged in the set $K(\bm{x}_a)$, therefore, we have inequality $\left( \frac{\partial L(\bm{x}_a)}{\partial x_{a1}}, \frac{\partial L(\bm{x}_a)}{\partial x_{a2}}, \cdots, \frac{\partial L(\bm{x}_a)}{\partial x_{ak}} \right) \cdot \Delta \bm{x} \leq 0$ that can be induced by the fact that $\sum_{h \in K(\bm{x}_a)}(x_{ah}' v_{bh})^\frac{1}{2} \leq \sum_{h \in K(\bm{x}_a)}(x_{ah} v_{bh})^\frac{1}{2}$.
    Furthermore we have that $\left(\frac{\partial \hat{u}_a(\bm{x}_a)}{\partial x_{aj}}\right)_{j \in K(\bm{x}_a)}$ is a linear combination of $\left( \frac{\partial L(\bm{x}_a)}{\partial x_{aj}}\right)_{j \in K(\bm{x}_a)}$ and $(1, 1, \cdots, 1)$. We take the partial derivative $\hat{u}_a(\bm{x}_a)$ and $L(\bm{x}_a)$ with regard to $x_{aj}$ for $j\in K(\bm{x}_a)$ respectively,
    \begin{equation*}
        \begin{aligned}
            \frac{\partial \hat{u}_a(\bm{x}_a)}{\partial x_{aj}} = &\frac{1}{2} \frac{v_{aj}}{\sqrt{v_{bj}x_{aj}}} \frac{\sum_{h \in K(\bm{x}_a)} (x_{ah} v_{bh})^\frac{1}{2}}{x_b + \sum_{h \in K(\bm{x}_a)} x_{ah}} + \left(\sum_{h \in K(\bm{x}_a)} \frac{v_{ah}\sqrt{x_{ah}}}{\sqrt{v_{bh}}} \right) \frac{\frac{1}{2}\frac{\sqrt{v_{bj}}}{\sqrt{x_{aj}}}(x_b + \sum_{h \in K(\bm{x}_a)} x_{ah}) - \sum_{h \in K(\bm{x}_a)} (x_{ah} v_{bh})^\frac{1}{2}}{(x_b + \sum_{h \in K(\bm{x}_a)} x_{ah})^2}  \\
            = &\left( \frac{1}{2} \frac{\sum_{h \in K(\bm{x}_a)} (x_{ah} v_{bh})^\frac{1}{2}}{x_b + \sum_{h \in K(\bm{x}_a)} x_{ah}} \frac{v_{aj}}{v_{bj}} + \frac{1}{2} \frac{\sum_{h \in K(\bm{x}_a)} \frac{v_{ah}\sqrt{x_{ah}}}{\sqrt{v_{bh}}}}{x_b + \sum_{h \in K(\bm{x}_a)} x_{ah}} \right) \frac{\sqrt{v_{bj}}}{\sqrt{x_{aj}}} - \frac{\sum_{h \in K(\bm{x}_a)} (x_{ah} v_{bh})^\frac{1}{2} \sum_{h \in K(\bm{x}_a)} \frac{v_{ah}\sqrt{x_{ah}}}{\sqrt{v_{bh}}}}{(x_b + \sum_{h \in K(\bm{x}_a)} x_{ah})^2},  \\
        \end{aligned}
    \end{equation*}
    $$\frac{\partial L(\bm{x}_a)}{\partial x_{aj}} = \frac{\partial \sum_{h \in K(\bm{x})} (x_{ah} v_{bh})^\frac{1}{2}}{\partial x_{aj}} = \frac{1}{2} \frac{\sqrt{v_{bj}}}{\sqrt{x_{aj}}}.$$
     For convenience, we define
    \begin{equation*}
        \begin{cases}
            C_1 = \frac{1}{2} \frac{\sum_{h \in K(\bm{x}_a)} (x_{ah} v_{bh})^\frac{1}{2}}{x_b + \sum_{h \in K(\bm{x}_a)} x_{ah}};  \\
            C_2 = \frac{1}{2} \frac{\sum_{h \in K(\bm{x}_a)} \frac{v_{ah}\sqrt{x_{ah}}}{\sqrt{v_{bh}}}}{x_b + \sum_{h \in K(\bm{x}_a)} x_{ah}};  \\
            C_3 = \frac{\sum_{h \in K(\bm{x}_a)} (x_{ah} v_{bh})^\frac{1}{2} \sum_{h \in K(\bm{x}_a)} \frac{v_{ah}\sqrt{x_{ah}}}{\sqrt{v_{bh}}}}{(x_b + \sum_{h \in K(\bm{x}_a)} x_{ah})^2},
        \end{cases}
    \end{equation*}
    we have
    \begin{equation}\label{equation-partial-ua-xa}
        \frac{\partial \hat{u}_a(\bm{x}_a)}{\partial x_{aj}} = (C_1 \frac{v_{aj}}{v_{bj}} + C_2) \frac{\sqrt{v_{bj}}}{\sqrt{x_{aj}}} - C_3.
    \end{equation}
    $\left(\frac{\partial \hat{u}_a(\bm{x}_a)}{\partial x_{aj}}\right)_{j \in K(\bm{x}_a)}$ is a linear combination of $\left( \frac{\partial L(\bm{x}_a)}{\partial x_{aj}}\right)_{j \in K(\bm{x}_a)}$ and $(1, 1, \cdots, 1)$, we have
    \begin{equation}\label{equation-linear-combination}
        \frac{\partial \hat{u}_a(\bm{x}_a)}{\partial x_{aj}} = C_4 \frac{\partial L(\bm{x}_a)}{\partial x_{aj}} + C_5.
    \end{equation}
    We substitute $\frac{\partial L(\bm{x}_a)}{\partial x_{aj}} = \frac{1}{2} \frac{\sqrt{v_{bj}}}{\sqrt{x_{aj}}}$ into Equation \eqref{equation-linear-combination} and obtain
    \begin{equation*}
        \frac{\partial \hat{u}_a(\bm{x}_a)}{\partial x_{aj}} = \frac{1}{2} C_4 \frac{\sqrt{v_{bj}}}{\sqrt{x_{aj}}} + C_5
    \end{equation*}
    Therefore, we have
    \begin{equation}\label{equation-partial-combination}
        (C_1 \frac{v_{aj}}{v_{bj}} + C_2) \frac{\sqrt{v_{bj}}}{\sqrt{x_{aj}}} - C_3 = C_4 \frac{\partial L(\bm{x}_a)}{\partial x_{aj}} + C_5.
    \end{equation}
    Arranging Equation \eqref{equation-partial-combination} can yield
    \begin{equation}
        \frac{v_{aj}}{v_{bj}} \frac{\sqrt{v_{bj}}}{\sqrt{x_{aj}}} = \frac{C_4 - C_2}{C_1} \frac{\sqrt{v_{bj}}}{\sqrt{x_{aj}}} + \frac{C_5 + C_3}{C_1}
    \end{equation}
    Let $\alpha = \frac{C_4 - C_2}{C_1}$ and $\beta = \frac{C_5 + C_3}{C_1}$, we have $\frac{v_{aj}}{v_{bj}} \frac{\sqrt{v_{bj}}}{\sqrt{x_{aj}}} = \alpha \frac{\sqrt{v_{bj}}}{\sqrt{x_{aj}}} + \beta$. Rearranging it, we can obtain $\frac{v_{aj}}{\sqrt{v_{bj}}} - \alpha \sqrt{v_{bj}} = \sqrt{x_{aj}} \beta$ for $\forall j \in K$.
\end{proof}

{\noindent \bf The Proof of Theorem \ref{sec4-theorem-optimal-commitment}}
\begin{proof}
    For the sake of simplicity, we define the following expressions.
    \begin{define}
        For $i \in \{a, b\}$, $K \subseteq [n]$, let
        \begin{align*}
            v_{iK} &= \sum_{j \in K} v_{ij}, \\
            c_K &= \sum_{j \in K} \frac{v_{aj}^2}{v_{bj}},  \\
            B_1(K) &= x_a v_{bK}^2 - 2 x_b v_{bK} v_{b\overline{K}}, \\
            B_2(K) &= 4 x_b v_{aK} v_{b\overline{K}} - 2 x_a v_{aK} v_{bK},  \\
            B_3(K) &= x_a v_{aK}^2 - 2 x_b v_{b\overline{K}} c_K,  \\
            B_4(K) &= x_a^2 v_{bK}^2 - 4 x_b (x_a + x_b) v_{bK} v_{b\overline{K}},  \\
            B_5(K) &= 8 x_b (x_a + x_b) v_{aK} v_{b\overline{K}} - 2 x_a^2 v_{aK} v_{bK},  \\
            B_6(K) &= x_a^2 v_{aK}^2 - 4 x_b (x_a + x_b) c_K v_{b\overline{K}},  \\
            \phi_1(\theta, K) &= B_1(K) \theta^2 + B_2(K) \theta + B_3(K),  \\
            \phi_2(\theta, K) &= B_4(K) \theta^2 + B_5(K) \theta + B_6(K).
        \end{align*}
    \end{define}

    We consider three cases, namely 
    \begin{itemize}
        \item {\bf Case 1.} $\forall j, h \in K$, $\frac{v_{aj}}{v_{bj}} = \frac{v_{ah}}{v_{bh}}$, or $K = \{1\}$.
        \item {\bf Case 2.1.} there exist $j, h \in K$, such that $\frac{v_{aj}}{v_{bj}} \neq \frac{v_{ah}}{v_{bh}}$, and $K = [n]$.
        \item {\bf Case 2.2.} there exist $j, h \in K$, such that $\frac{v_{aj}}{v_{bj}} \neq \frac{v_{ah}}{v_{bh}}$, and $K \subsetneq [n]$.
    \end{itemize}
    
    Let $\bm{x}_a$ denote the leader's optimal commitment and $\bm{x}_b$ denote the follower's best response.
    First, we consider {\bf Case 1}.
    Let $x_{aK} = \sum_{l \in K} x_{al}$.
    By Theorem \ref{theorem-characterization}, it is easy to obtain that for $\forall j \in K$, $x_{aj} = x_{aK} \frac{v_{aj}}{v_{aK}}$.
    For $\forall j \in \overline{K}$, according to Lemma \ref{sec4-lemma-leader-threshold}, we have
    \begin{align*}
        x_{aj} = \frac{v_{bj}(x_b + x_{aK})^2}{\left(\sum_{l \in K} (x_{al}v_{bl})^\frac{1}{2}\right)^2} &= \frac{v_{bj}(x_b + x_{aK})^2}{\left(\sum_{l \in K} (x_{aK} \frac{v_{al}}{v_{aK}}  v_{bl})^\frac{1}{2}\right)^2} = \frac{v_{bj}(x_b + x_{aK})^2}{x_{aK}v_{bK}}.
    \end{align*}

    The amount of resources allocated by the leader in the set $\overline{K}$ is given by $\sum_{j \in \overline{K}} x_{aj} = \sum_{j \in \overline{K}} \frac{v_{bj}(x_b + x_{aK})^2}{x_{aK}v_{bK}}$.
    Then, we have
    \begin{align}\label{sec4-equ-case1}
        x_{aK} + \sum_{j \in \overline{K}} \frac{v_{bj}(x_b + x_{aK})^2}{x_{aK}v_{bK}} = x_a.
    \end{align}
    Rearranging Equation \eqref{sec4-equ-case1}, we obtain
    \begin{align*}
        (v_{b\overline{K}} + v_{bK}) (x_{aK})^2 + (2 x_b v_{b\overline{K}} - x_a v_{bK}) x_{aK} + x_b^2v_{b\overline{K}} = 0.
    \end{align*}
    This is a quadratic equation. 
    It has two positive real roots for $x_{aK}$.
    In this case, we choose the larger root to maximize the leader's utility:
    \begin{equation*}   
        \begin{aligned}
            &x_{aK} = &\frac{x_a v_{bK} - 2x_b v_{b\overline{K}} + \sqrt{x_a^2 v_{bK}^2 - 4x_a x_b v_{bK} v_{b\overline{K}} - 4x_b^2 v_{bK} v_{b\overline{K}}}}{2(v_{b\overline{K}} + v_{bK})}.
        \end{aligned}
    \end{equation*}
    Obviously, it requires that $x_a^2 v_{bK}^2 - 4x_a x_b v_{bK} v_{b\overline{K}} - 4x_b^2 v_{bK} v_{b\overline{K}} \geq 0$.
    If $x_a^2 v_{bK}^2 - 4x_a x_b v_{bK} v_{b\overline{K}} - 4x_b^2 v_{bK} v_{b\overline{K}} < 0$, it means that for the given set $K$, the leader does not have sufficient resources to cover the battlefield sets $\overline{K}$ or $K$.
    In this case, we should seek the next $K$.
    Therefore, if $x_a^2 v_{bK}^2 - 4x_a x_b v_{bK} v_{b\overline{K}} - 4x_b^2 v_{bK} v_{b\overline{K}} \geq 0$, the optimal commitment is as follows
    \begin{align*}
        x_{aj} = \begin{cases}
            x_{aK} \cdot \frac{v_{aj}}{v_{aK}}, &\text{if} \; j \in K;  \\
            \frac{v_{bj} \left( x_b + x_{aK} \right)^2}{\left( \sum_{l \in K} (x_{al} v_{bl})^{\frac{1}{2}} \right)^2}, &\text{if} \; j \in \overline{K}.
        \end{cases}
    \end{align*}
    This can be calculated in $O(n)$ time.

    Second, we consider the case where $\frac{v_{aj}}{v_{bj}} \neq \frac{v_{ah}}{v_{bh}}$, leading to cases {\bf 2.1} and {\bf 2.2}. In this case, $\beta \neq 0$.
    According to Lemma \ref{sec4-lemma-parameters}, for $\forall j \in K$, we have $(\frac{v_{aj}}{\sqrt{v_{bj}}} - \alpha v_{bj}^\frac{1}{2})\frac{1}{\beta} = x_{aj}^\frac{1}{2}$.
    The resources invested by the leader in $K$ are $\frac{1}{\beta^2} \sum_{j \in K}(\frac{v_{aj}}{\sqrt{v_{bj}}} - \alpha \sqrt{v_{bj}})^2$.
    For $\forall j \in \overline{K}$, we have
    \begin{equation*}
        \begin{aligned}
            x_{aj} &= v_{bj} \cdot \frac{(x_b + x_K)^2}{\left( \sum_{h \in K}(v_{ah} - \alpha v_{bh})\frac{1}{\beta} \right)^2},
        \end{aligned}
    \end{equation*}
    where $x_K = \sum_{j \in K} x_{aj} = \frac{1}{\beta^2} \sum_{j \in K} (\frac{v_{aj}}{\sqrt{v_{bj}}} - \alpha v_{bj}^\frac{1}{2})^2$.
    The resources invested by the leader in $\overline{K}$ are
    \begin{align*}
        \sum_{j \in \overline{K}} x_{aj} = \sum_{j \in \overline{K}} v_{bj} \cdot \beta^2 \frac{(x_b + x_K)^2}{\left(\sum_{h \in K} (v_{ah} - \alpha v_{bh}) \right)^2}.
    \end{align*}
    Let us analyze {\bf Case 2.1}.
    Because of $K = [n]$, we have $\frac{1}{\beta^2} \sum_{j \in [n]}(\frac{v_{aj}}{\sqrt{v_{bj}}} - \alpha \sqrt{v_{bj}})^2 = x_a$.
    It follows that $\beta^2 = \frac{\sum_{j \in [n]}(\frac{v_{aj}}{\sqrt{v_{bj}}} - \alpha \sqrt{v_{bj}})^2}{x_a}$.
    According to Problem \eqref{rewrite-optimal-commitment}, we have
    \begin{equation}\label{sec4-problem-subcase1-optimal-commitment}
        \begin{aligned}
            \max_{\bm{x}_a \in \bm{X}_a} \quad & \hat{u}_a(\bm{x}_a) = \frac{ \left(\sum_{j \in [n]}(\frac{v_{aj}}{v_{bj}} - \alpha)v_{aj}\right) \left(\sum_{j \in [n]} (v_{aj} - \alpha v_{bj})\right)}{\frac{\sum_{j \in [n]}(\frac{v_{aj}}{\sqrt{v_{bj}}} - \alpha \sqrt{v_{bj}})^2}{x_a} x_b + \sum_{j \in [n]}(\frac{v_{aj}}{\sqrt{v_{bj}}} - \alpha \sqrt{v_{bj}})^2} &  \\
            \mbox{s.t.} \quad & \frac{\sum_{j \in [n]}(\frac{v_{aj}}{\sqrt{v_{bj}}} - \alpha \sqrt{v_{bj}})^2}{x_a} > 0. &
        \end{aligned}
    \end{equation}
    We need to find a suitable parameter $\alpha$ that maximizes the leader's utility $\hat{u}_a(\bm{x}_a)$. Therefore, the optimal commitment is as follows
    \begin{align}\label{sec4-equa-subcase1-optimal-commitment}
        x_{aj} = \frac{(\frac{v_{aj}}{\sqrt{v_{bj}}} - \alpha \sqrt{v_{bj}})^2 \cdot x_a}{\sum\limits_{j \in [n]}(\frac{v_{aj}}{\sqrt{v_{bj}}} - \alpha \sqrt{v_{bj}})^2}, \; \forall j \in [n].
    \end{align}
    To find a suitable parameter $\alpha$, we rewrite the Problem \ref{sec4-problem-subcase1-optimal-commitment} as follows
    \begin{equation}\label{sec4-problem-subcase1-optimal-commitment-rewrite}
        \max\limits_{\bm{x}_a \in \bm{X}_a} \hat{u}_a(\bm{x}_a) = \frac{\left((\sum\limits_{j \in [n]} v_{aj})^2 - (\sum\limits_{j \in [n]} v_{bj}) (\sum\limits_{j \in [n]} \frac{v_{aj}^2}{v_{bj}})\right)\alpha}{\alpha^2 (\sum\limits_{j \in [n]} v_{bj}) - 2 \alpha (\sum\limits_{j \in [n]} v_{aj}) + \sum\limits_{j \in [n]} \frac{v_{aj}^2}{v_{bj}}}
    \end{equation}
    Note that $(\sum_{j \in [n]} v_{aj})^2 - (\sum_{j \in [n]} v_{bj}) (\sum_{j \in [n]} \frac{v_{aj}^2}{v_{bj}}) < 0$, we choose the $\alpha$ with negative value to maximize the leader's utility.
    Solve the Problem \ref{sec4-problem-subcase1-optimal-commitment-rewrite}, we have $\alpha = -\sqrt{\frac{1}{\sum_{j \in [n]} v_{bj}} \sum_{j \in [n]} \frac{v_{aj}^2}{v_{bj}}}$.
    Substitute $\alpha$ into Equation \eqref{sec4-equa-subcase1-optimal-commitment}, the closed-form solution is obtained, that is
    \begin{equation*}
        x_{aj} = \frac{ \left( \frac{v_{aj}}{\sqrt{v_{bj}}} + \left( \frac{\sum_{h = 1}^{n}\frac{(v_{ah})^2}{v_{bh}}}{\sum_{h = 1}^{n} v_{bh}} \right)^\frac{1}{2} \sqrt{v_{bj}} \right)^2} {\sum_{l = 1}^{n} \left( \frac{v_{al}}{\sqrt{v_{bl}}} + \left( \frac{\sum_{h = 1}^{n}\frac{(v_{ah})^2}{v_{bh}}}{\sum_{h = 1}^{n} v_{bh}} \right)^\frac{1}{2} \sqrt{v_{bl}} \right)^2} \cdot x_a.
    \end{equation*}
    This can be calculated in $O(n)$ time.

    Finally, let us analyse {\bf Case 2.2}.
    Because of $K \subsetneq [n]$, we can derive following equation
    \begin{equation}\label{sec4-equation-constrain}
        \begin{aligned}
            x_a = &\frac{1}{\beta^2} \sum_{j \in K}(\frac{v_{aj}}{\sqrt{v_{bj}}} - \alpha \sqrt{v_{bj}})^2 + (\sum_{j \in \overline{K}} v_{bj}) \frac{\left[\beta x_b + \frac{1}{\beta} \sum_{j \in K} (\frac{v_{aj}}{\sqrt{v_{bj}}} - \alpha \sqrt{v_{bj}})^2\right]^2}{\left(\sum_{h \in K} (v_{ah} - \alpha v_{bh})\right)^2}.
        \end{aligned}
    \end{equation}
    Let $y = \beta^2$, we have $x_{aj}^\frac{1}{2} = (\frac{v_{aj}}{\sqrt{v_{bj}}} - \alpha v_{bj}^\frac{1}{2})\frac{1}{\beta} = (\frac{v_{aj}}{\sqrt{v_{bj}}} - \alpha v_{bj}^\frac{1}{2})\frac{1}{\sqrt{y}}$.
    Substituting $x_{aj}^\frac{1}{2} = (\frac{v_{aj}}{\sqrt{v_{bj}}} - \alpha v_{bj}^\frac{1}{2})\frac{1}{\sqrt{y}}$ into the objective function of Problem \ref{rewrite-optimal-commitment}, we obtain the leader's utility, that is
    \begin{equation}\label{sec4-equation-subcase2-optimal-utility-solve}
        \begin{aligned}
            \hat{u}_a(\bm{x}_a) &= \left( \sum_{j \in K} (\frac{v_{aj}}{v_{bj}} - \alpha) \frac{1}{\sqrt{y}} v_{aj} \right) \frac{\sum_{j \in K} (v_{aj} - \alpha v_{bj})\frac{1}{\sqrt{y}}}{\frac{1}{y} \sum_{j \in K} (\frac{v_{aj}}{\sqrt{v_{bj}}} - \alpha v_{bj}^\frac{1}{2})^2 + x_b} \\
            &= \frac{1}{y} \cdot \frac{ \left(\sum_{j \in K}(\frac{v_{aj}}{v_{bj}} - \alpha)v_{aj}\right) \left(\sum_{j \in K} (v_{aj} - \alpha v_{bj})\right)}{x_b + \frac{1}{y} \sum_{j \in K}(\frac{v_{aj}}{\sqrt{v_{bj}}} - \alpha \sqrt{v_{bj}})^2}  \\
            &= \frac{ \left(\sum_{j \in K}(\frac{v_{aj}}{v_{bj}} - \alpha)v_{aj}\right) \left(\sum_{j \in K} (v_{aj} - \alpha v_{bj})\right)}{y x_b + \sum_{j \in K} (\frac{v_{aj}}{\sqrt{v_{bj}}} - \alpha \sqrt{v_{bj}})^2} 
        \end{aligned}
    \end{equation}
    Obviously, $\hat{u}_a(\bm{x}_a)$ increases as $y$ decreases for a given $\alpha$. We can rewrite Equation \eqref{sec4-equation-constrain} as
    \begin{equation}\label{sec4-equation-rewrite-constrain-y}
        \begin{aligned}
            x_a = y^{-1} \sum_{j \in K}(\frac{v_{aj}}{\sqrt{v_{bj}}} - \alpha \sqrt{v_{bj}})^2 + (\sum_{j \in \overline{K}} v_{bj}) \frac{y \left(x_b + y^{-1} \sum_{j \in K} (\frac{v_{aj}}{\sqrt{v_{bj}}} - \alpha \sqrt{v_{bj}})^2\right)^2}{\left(\sum_{h \in K} (v_{ah} - \alpha v_{bh})\right)^2}.
        \end{aligned}
    \end{equation}
    For convenience, we define
    \begin{equation*}
        \begin{cases}
            A_1 = \sum_{j \in K}(\frac{v_{aj}}{\sqrt{v_{bj}}} - \alpha \sqrt{v_{bj}})^2  \\
            A_2 = \frac{\sum_{j \in \overline{K}} v_{bj}}{\left(\sum_{h \in K} (v_{ah} - \alpha v_{bh})\right)^2}  \\
        \end{cases}
    \end{equation*}
    Therefore, we have
    \begin{equation*}
        y^{-1} A_1 + A_2 y \left( x_b + y^{-1} A_1 \right)^2  = x_a.
    \end{equation*}
    Furthermore, it can be rearranged as
    \begin{align*}
        A_2 x_b^2 y^2 + (2 A_1 A_2 x_b - x_a) y + A_1 + A_2 A_1^2 = 0,
    \end{align*}
    which is a quadratic equation about $y$.
    Let $y_1$ and $y_2$ be roots of this equation, respectively.
    We have $y_1 y_2 > 0$, which implies $y_1 > 0$ and $y_2 > 0$.
    For a given $\alpha$, we choose a smaller root $\hat{y} = \min\{y_1, y_2\}$ to maximize the leader's utility.
    Thus, we have
    \begin{equation}
        \begin{aligned}
            \hat{y} = & \frac{(x_a - 2 A_1 A_2 x_b) - \sqrt{(2 A_1 A_2 x_b - x_a)^2 - 4 A_2 x_b^2 (A_1 + A_2 A_1^2)}}{2 A_2 x_b^2}  \\
            =& \frac{(x_a - 2 A_1 A_2 x_b) - \sqrt{x_a^2 - 4 A_1 A_2 x_a x_b - 4 A_1 A_2 x_b^2}}{2 A_2 x_b^2}  \\
            =& \frac{\left(x_a - 2 \frac{\left(\sum_{j \in K}(\frac{v_{aj}}{\sqrt{v_{bj}}} - \alpha \sqrt{v_{bj}})^2 \right) (\sum_{j \in \overline{K}} v_{bj})}{\left(\sum_{h \in K} (v_{ah} - \alpha v_{bh})\right)^2} x_b\right)}{2 \frac{\sum_{j \in \overline{K}} v_{bj}}{\left(\sum_{h \in K} (v_{ah} - \alpha v_{bh})\right)^2} x_b^2} - \frac{\sqrt{x_a^2 - 4 \frac{\left(\sum_{j \in K}(\frac{v_{aj}}{\sqrt{v_{bj}}} - \alpha \sqrt{v_{bj}})^2 \right) (\sum_{j \in \overline{K}} v_{bj})}{\left(\sum_{h \in K} (v_{ah} - \alpha v_{bh})\right)^2} x_b (x_a + x_b)}}{2 \frac{\sum_{j \in \overline{K}} v_{bj}}{\left(\sum_{h \in K} (v_{ah} - \alpha v_{bh})\right)^2} x_b^2}
            \end{aligned}
    \end{equation}
    After simplification, we obtain
    \begin{align*}
        \hat{y} = \frac{\phi_1(\alpha, K) - (v_{aK} - v_{bK} \alpha) \sqrt{\phi_2(\alpha, K)}}{2 x_b^2 v_{b\overline{K}}}
    \end{align*}
    When $\phi_2(\alpha, K) \geq 0$, $y$ has real number solutions.
    Substitute $\hat{y}$ into Equation \eqref{sec4-equation-subcase2-optimal-utility-solve}, and obtain the following optimization problem
    \begin{equation}\label{sec4-problem-beta-neq0-subset-commitment}
        \begin{aligned}
            \max_{\bm{x_a} \in \bm{X}_a} \quad &\hat{u}(\bm{x}_a) = \frac{ \left(\sum_{j \in K}(\frac{v_{aj}}{v_{bj}} - \alpha)v_{aj}\right) \left(\sum_{j \in K} (v_{aj} - \alpha v_{bj})\right)}{\hat{y} x_b + \sum_{j \in K}(\frac{v_{aj}}{\sqrt{v_{bj}}} - \alpha \sqrt{v_{bj}})^2} &   \\
            \mbox{s.t.} \quad & x_a^2 - 4 \frac{\left(\sum_{j \in K}(\frac{v_{aj}}{\sqrt{v_{bj}}} - \alpha \sqrt{v_{bj}})^2 \right) (\sum_{j \in \overline{K}} v_{bj})}{\left(\sum_{h \in K} (v_{ah} - \alpha v_{bh})\right)^2} x_b (x_a + x_b) \geq 0. &
        \end{aligned}
    \end{equation}
    We need to find a suitable parameter $\alpha$ that maximizes the leader's utility $\hat{u}_a(\bm{x}_a)$. Therefore, the optimal commitment is as follows
    \begin{align*}
        x_{aj} = \begin{cases}
            \frac{2 x_b^2 v_{b\overline{K}} (\frac{v_{aj}}{\sqrt{v_{bj}}} - \alpha \sqrt{v_{bj}})^2}{\phi_1(\alpha, K) - (v_{aK} - v_{bK} \alpha) \sqrt{\phi_2(\alpha, K)}}, \; j \in K;  \\
            \frac{v_{bj} \left( x_b + \sum_{l \in K} x_{al} \right)^2}{\left( \sum_{l \in K} (x_{al} v_{bl})^{\frac{1}{2}} \right)^2}, \; j \in \overline{K}.
        \end{cases}
    \end{align*}
    This can be calculated in $O(n)$ time.
\end{proof}

{\noindent \bf Proof of Lemma \ref{sec4-lemma-allocate-positive-resources-allbattlefields}}
\begin{proof}
    We prove it by contradiction.
    Suppose the leader does not allocate resources in some battlefields.
    Let $Q = \{1, 2, \cdots, q \}$ be a battlefield set where the leader allocates positive resources, $q < n$.
    Let $\overline{Q} = \{q+1, \cdots, n \}$ be a battlefield set where the leader does not allocate resources.
    Without loss of generality, we assume that $\frac{v_{b1}}{x_{a1}} \geq \cdots \geq \frac{v_{bq}}{x_{aq}}$.
    Let $K$ denote the battlefield set where the follower allocates positive resources in his best response, $|K| = k$.
    Then we have $\frac{v_{bk}}{x_{ak}} > \frac{\left(\sum_{l \in K} (x_{al} v_{bl})^\frac{1}{2}\right)^2}{(x_b + \sum_{l \in K} x_{al})^2}$ according to Lemma \ref{sec2-lemma-br}.

    If $k+1 \leq q$, we have
    \begin{align*}
        \frac{v_{b(k+1)}}{x_{a(k+1)}} \leq \frac{\left(\sum_{l \in K}(x_{al} v_{bl})^\frac{1}{2} + (x_{a(k+1)}^\frac{1}{2} v_{b(k+1)}^\frac{1}{2})\right)^2}{(x_b + \sum_{l \in K} x_{al} + x_{a(k+1)})^2}.
    \end{align*}
    We define the another commitment $\hat{\bm{x}}_a(\varepsilon)$ such that $\hat{x}_{al}(\varepsilon) = (1 - \varepsilon) x_{al}$ for $1 \leq l \leq k$, $\hat{x}_{an}(\varepsilon) = \varepsilon \sum_{l \in K} x_{al}$, and $\hat{x}_{al}(\varepsilon) = x_{al}$ for $k+1 \leq l \leq n-1$.
    When $\varepsilon$ is a small enough positive number, according to Lemma \ref{sec2-lemma-br}, we can easily verify that the following inequality holds
    \begin{align*}
        \frac{v_{bk}}{\hat{x}_{ak}(\varepsilon)} > \frac{\left(\sum_{l \in K}(\hat{x}_{al}(\varepsilon) v_{bl})^\frac{1}{2} + (\hat{x}_{an}(\varepsilon) v_{bn})^\frac{1}{2}\right)^2}{\left(x_b + \sum_{l \in K} \hat{x}_{al}(\varepsilon) + \hat{x}_{an}(\varepsilon)\right)^2}.
    \end{align*}
    It implies that $K(\hat{\bm{x}}_a(\varepsilon)) = \{ n \} \cup K(\bm{x}_a)$.
    Thus, we have $\frac{v_{bn}}{\hat{x}_{an}(\varepsilon)} \geq \frac{v_{b1}}{\hat{x}_{a1}(\varepsilon)} \geq \cdots \geq \frac{v_{bk}}{\hat{x}_{ak}(\varepsilon)}$. On the other hand, because of $k+1 \leq q$, we can easily verify that
    \begin{equation*}
        \begin{aligned}
            \frac{v_{b(k+1)}}{\hat{x}_{a(k+1)}(\varepsilon)} \leq \frac{\left(\sum_{l \in K}(\hat{x}_{al}(\varepsilon) v_{bl})^\frac{1}{2} + (\hat{x}_{an}(\varepsilon) v_{bn})^\frac{1}{2} + (\hat{x}_{a(k+1)}(\varepsilon) v_{b(k+1)})^\frac{1}{2} \right)^2}{\left(x_b + \sum_{l \in K} \hat{x}_{al}(\varepsilon) + \hat{x}_{an}(\varepsilon)\right)^2}.
        \end{aligned}
    \end{equation*}
    Finally, we prove that the leader's utility increases when the leader makes commitment $\hat{\bm{x}}_a(\varepsilon)$ compared to $\bm{x}_a$.
    \begin{equation*}
        \begin{aligned}
            u_a(\hat{\bm{x}}_a(\varepsilon)) = &\left(\sum_{j \in K} (\frac{\hat{x}_{aj}(\varepsilon)}{v_{bj}})^\frac{1}{2} v_{aj} + (\frac{\hat{x}_{an}(\varepsilon)}{v_{bn}})^\frac{1}{2} v_{an} \right) \times \left(\frac{\sum_{l \in K} (\hat{x}_{al}(\varepsilon) v_{bl})^\frac{1}{2} + (\hat{x}_{an}(\varepsilon) v_{bn})^\frac{1}{2}}{\sum_{l \in K} \hat{x}_{al}(\varepsilon) + \hat{x}_{an}(\varepsilon) + x_b}\right) \\
            = & \left((1 - \varepsilon)^\frac{1}{2} \sum_{j \in K} (\frac{x_{aj}}{v_{bj}})^\frac{1}{2} v_{aj} + \frac{(\varepsilon x_{aK})^\frac{1}{2}}{v_{bn}^\frac{1}{2}} v_{an} \right) \times \left( \frac{(1 - \varepsilon)^\frac{1}{2} \sum_{j \in K} (x_{aj} v_{bj})^\frac{1}{2} + \varepsilon^\frac{1}{2} x_{aK}^\frac{1}{2} v_{bn}^\frac{1}{2}}{x_K + x_b}  \right)  \\
            = & \frac{1}{x_{aK} + x_b} \Bigg[ (1-\varepsilon) \sum_{j \in K} (\frac{x_{aj}}{v_{bj}})^\frac{1}{2} v_{aj} \sum_{j \in K} (x_{aj} v_{bj})^\frac{1}{2} + \varepsilon x_{aK} v_{an} + (\varepsilon - \varepsilon^2)^\frac{1}{2} \Big[(x_{aK} v_{bn})^\frac{1}{2} \sum_{j \in K} (\frac{x_{aj}}{v_{bj}})^\frac{1}{2} v_{aj} \Big] + \\
            &(\varepsilon - \varepsilon^2)^\frac{1}{2} (\frac{x_{aK}}{v_{bn}})^\frac{1}{2} v_{an} \sum_{j \in K} (x_{aj} v_{bj})^\frac{1}{2} \Bigg]
        \end{aligned}
    \end{equation*}
    where $x_{aK} = \sum_{j \in K} x_{aj}$.
    The change of the leader's utility is that
    \begin{align*}
        & u_a(\hat{\bm{x}}_a(\varepsilon)) - u_a(\bm{x}_a) \\
        = &\frac{\varepsilon}{x_{aK} + x_b} \Bigg[- \sum_{j \in K} (\frac{x_{aj}}{v_{bj}})^\frac{1}{2} \sum_{j \in K} (x_{aj} v_{bj})^\frac{1}{2} + x_{aK} v_{an} + (\frac{1 - \varepsilon}{\varepsilon})^\frac{1}{2} \Big[(x_{aK} v_{bn})^\frac{1}{2} \sum_{j \in K} (\frac{x_{aj}}{v_{bj}})^\frac{1}{2} + (\frac{x_{aK}}{v_{bn}})^\frac{1}{2} v_{an} \sum_{j \in K} (x_{aj} v_{bj})^\frac{1}{2}
        \Big] \Bigg]
    \end{align*}
    When $\varepsilon$ is small enough, we have $u_a(\hat{\bm{x}}_a(\varepsilon)) - u_a(\bm{x}_a) > 0$.
    If there is a battlefield where the leader does not allocate resources, then the leader can transfer $\varepsilon$ resources to that battlefield, thereby increasing the utility of the leader. The leader allocates positive resources across all battlefields.
\end{proof}

\section{OMITTED PROOFS FROM SECTION 5}\label{appendix-section5}
{\bf Proof of Lemma \ref{sec5-lemma-nece-suff}}
\begin{proof}
    Previous works have established that if the strategy profile $(\bm{x}_a, \bm{x}_b)$ is a NE, then $x_{ij} > 0$, for $\forall i \in \{a, b\}$, $\forall j \in [n]$ \citep{Kim-2018,Li-2022}.
    Hence, if $(\bm{x}_a, BR(\bm{x}_a))$ constitutes a NE, we must have $K(\bm{x}_a) = [n]$.
    Therefore, we examine the situation where the leader's commitment results in the follower having positive best responses across all battlefields.

    According to Equation \eqref{sec2-lem-equ-br}, if $\hat{\bm{x}}_b$ is the best response to $\hat{\bm{x}}_a$ and $\hat{\bm{x}}_a$ is also the best response to $\hat{\bm{x}}_b$, then we have
    \begin{equation}\label{sec5-equation-NE-xa-xb}
        \begin{cases}
            \hat{\bm{x}}_b = \left( \frac{(\hat{x}_{aj} v_{bj})^{\frac{1}{2}} (x_b + x_a)}{\sum_{j' = 1}^n (\hat{x}_{aj'} v_{bj'})^{\frac{1}{2}}} - \hat{x}_{aj} \right)_{j = 1, 2, \cdots, n};  \\
            \hat{\bm{x}}_a = \left( \frac{(\hat{x}_{bj} v_{aj})^\frac{1}{2} (x_a + x_b)}{\sum_{k = 1}^{n}(\hat{x}_{bk} v_{ak})^\frac{1}{2}} - \hat{x}_{bj} \right)_{j = 1, 2, \cdots, n}.
        \end{cases}
    \end{equation}
    By Equation \eqref{sec5-equation-NE-xa-xb}, we have the following equation
    \begin{equation}\label{sec4-equation-optimal-NE-conditon}
        \begin{aligned}
            \frac{(\hat{x}_{aj} v_{bj})^\frac{1}{2} (x_a + x_b)}{\sum_{k = 1}^{n}(\hat{x}_{ak} \cdot v_{bk})^\frac{1}{2}} = \hat{x}_{aj} + \hat{x}_{bj} = \frac{(\hat{x}_{bj} v_{aj})^\frac{1}{2} (x_a + x_b)}{\sum_{k = 1}^{n}(\hat{x}_{bk} \cdot v_{ak})^\frac{1}{2}}, \forall j \in [n].
        \end{aligned}
    \end{equation}
    
    In Equation \eqref{sec4-equation-optimal-NE-conditon}, the first equality holds by using Lemma \ref{sec2-lemma-br} and the fact that $\hat{\bm{x}}_b$ is best response to $\hat{\bm{x}}_a$. The second equality holds by using Lemma \ref{sec2-lemma-br} and the fact that $\hat{\bm{x}}_a$ is also best response to $\hat{\bm{x}}_b$.

    By Equation \eqref{sec4-equation-optimal-NE-conditon}, we have $\frac{(\hat{x}_{aj} \cdot v_{bj})^\frac{1}{2}}{(\hat{x}_{bj} \cdot v_{aj})^\frac{1}{2}} = \frac{\sum_{k = 1}^{n}(\hat{x}_{ak} \cdot v_{bk})^\frac{1}{2}}{\sum_{k = 1}^{n}(\hat{x}_{bk} \cdot v_{ak})^\frac{1}{2}}$, implying $\forall j \in [n]$, $\frac{\hat{x}_{aj} \cdot v_{bj}}{\hat{x}_{bj} \cdot v_{aj}}$ is a constant. We substitute $\hat{x}_{bj}$ into Equation \eqref{sec4-equation-optimal-NE-conditon} yielding
    \begin{equation}\label{sec4-equation-optimal-NE-calcute-1}
        \begin{aligned}
            \frac{\hat{x}_{aj} \cdot v_{bj}}{\hat{x}_{bj} \cdot v_{aj}} &= \frac{\hat{x}_{aj} \cdot v_{bj}}{\left[ \frac{(\hat{x}_{aj} \cdot v_{bj})^\frac{1}{2} \cdot (x_a + x_b)}{\sum_{k = 1}^{n}(\hat{x}_{ak} \cdot v_{bk})^\frac{1}{2}} - \hat{x}_{aj} \right] \cdot v_{aj}} = \frac{1}{\frac{x_a + x_b}{\sum_{k = 1}^{n}(\hat{x}_{ak} v_{bk})^\frac{1}{2}} \cdot \frac{v_{aj}}{(\hat{x}_{aj}v_{bj})^\frac{1}{2}} - \frac{v_{aj}}{v_{bj}}}.
        \end{aligned}
    \end{equation}
    According to Lemma \ref{sec4-lemma-parameters}, we have $\sqrt{\hat{x}_{aj}} = \frac{1}{\beta} (\frac{v_{aj}}{\sqrt{v_{bj}}} - \sqrt{v_{bj}} \alpha)$. Let $\gamma = \frac{1}{\beta}$ and $r = -\alpha$, we have $\sqrt{\hat{x}_{aj}} = \gamma (\frac{v_{aj}}{\sqrt{v_{bj}}} + r \sqrt{v_{bj}})$.
    Substituting it into Equation \eqref{sec4-equation-optimal-NE-calcute-1} yields
   \begin{equation}\label{sec4-equation-optimal-NE-calcute-3}
    \begin{aligned}
            \frac{\hat{x}_{bj} \cdot v_{aj}}{\hat{x}_{aj} \cdot v_{bj}}  &= \frac{x_a + x_b}{\sum_{k = 1}^{n}(\hat{x}_{ak} \cdot v_{bk})^\frac{1}{2}} \cdot \frac{v_{aj}}{\gamma (\frac{v_{aj}}{\sqrt{v_{bj}}} + r \sqrt{v_{bj}}) (v_{bj})^\frac{1}{2}} - \frac{v_{aj}}{v_{bj}}  \\
            &= \frac{x_a + x_b}{\sum_{k = 1}^{n}(\hat{x}_{ak} \cdot v_{bk})^\frac{1}{2}} \cdot \frac{v_{aj}}{\gamma (v_{aj} + r v_{bj})} - \frac{v_{aj}}{v_{bj}}  \\
            &= \frac{x_a + x_b}{\sum_{k = 1}^{n}(\hat{x}_{ak} \cdot v_{bk})^\frac{1}{2}} \cdot \frac{1}{\gamma (1 + r \frac{v_{bj}}{v_{aj}})} - \frac{v_{aj}}{v_{bj}}.
    \end{aligned}
    \end{equation}
    Let $t = \frac{v_{aj}}{v_{bj}}$, $D = \frac{x_a + x_b}{\sum_{k = 1}^{n}(\hat{x}_{ak} \cdot v_{bk})^\frac{1}{2}}$, and let $E = \frac{\hat{x}_{bj} \cdot v_{aj}}{\hat{x}_{aj} \cdot v_{bj}}$, then Equation \eqref{sec4-equation-optimal-NE-calcute-3} can be rewritten as
    \begin{equation}\label{sec4-equation-optimal-NE-calcute-4}
        E = D \cdot \frac{1}{\gamma (1 + \frac{r}{t})} - t.
    \end{equation}
    We rearrange Equation \eqref{sec4-equation-optimal-NE-calcute-4} and obtain $\gamma t^2 + (\gamma E + \gamma r - D)t + \gamma E r = 0$. The terms $r, D, E, \gamma$ are constants independent of $t$, therefore we can view $t$ as a variable. Equation \eqref{sec4-equation-optimal-NE-calcute-4} is a quadratic equation with regard to $t$. Therefore, for any profile $(D, \gamma, r, E)$, there are at most two solutions for $t$ in the Equation \eqref{sec4-equation-optimal-NE-calcute-4}, which means that $t = \frac{v_{aj}}{v_{bj}}$ may have two different values at most such that the optimal commitment and its best response constitute a Nash equilibrium.
\end{proof}

{\noindent\bf Proof of Theorem \ref{finally-theorem}}
\begin{proof}
    First, for simplicity, let's define a few expressions.
    \begin{define}
        For any $T \subseteq [n]$, $i \in \{a,b\}$, let
        \begin{align*}
            v_{iT} &= \sum_{j \in T} v_{ij}, \\
            \Psi_1(\xi_1, \xi_2, \xi_3, \xi_4) &= \left( \xi_1 \xi_3^{-\frac{1}{2}} + \left( \frac{\xi_1^2 \xi_3^{-1} + \xi_2^2 \xi_4^{-1}}{\xi_3 + \xi_4}\right)^\frac{1}{2} \cdot \xi_3^\frac{1}{2} \right)^2, \\
            t_1(T) &= \Psi_1(v_{aT}, v_{a\overline{T}}, v_{bT}, v_{b\overline{T}}), \\
            t_2(T) &= \Psi_1(v_{a\overline{T}}, v_{aT}, v_{b\overline{T}}, v_{bT}), \\
            \Psi_2(\xi_1, \xi_2, \xi_3, \xi_4, \xi_5, \xi_6) &= \frac{\xi_6 \xi_4 \xi_1 (\xi_5 \xi_3)^\frac{1}{2}}{(\xi_5 \xi_3)^\frac{1}{2} + (\xi_6 \xi_4)^\frac{1}{2}} - \frac{\xi_5 \xi_3 \xi_2 (\xi_6 \xi_4)^\frac{1}{2}}{ (\xi_5 \xi_3)^\frac{1}{2} + (\xi_6 \xi_4)\frac{1}{2}}, \\
            \Psi_3(\xi_1, \xi_2, \xi_3, \xi_4, \xi_5, \xi_6) &= \frac{\xi_1 \xi_4 \xi_6 \xi_5}{\xi_5 + \xi_6} - \frac{\xi_6 \xi_5 \xi_3 \xi_2}{\xi_5 + \xi_6}, \\
            t_3(T) &= \Psi_2(v_{a\overline{T}}, v_{aT}, v_{b\overline{T}}, v_{bT}, t_2(T), t_1(T))  + \Psi_3(v_{aT}, v_{a\overline{T}}, v_{bT}, v_{b\overline{T}}, t_1(T), t_2(T)), \\
            t_4(T) &= \Psi_2(v_{aT}, v_{a\overline{T}}, v_{bT}, v_{b\overline{T}}, t_1(T), t_2(T)), \\
            f(v_{aT}, v_{bT}, v_{a\overline{T}}, v_{b\overline{T}}) &= \frac{t_4(T)}{t_3(T)}
        \end{align*}
    \end{define}
    Let $\bm{x}_a$ denote the leader's optimal commitment and $\bm{x}_b$ denote the follower's best response.
    For the first case, it is easy to verify that $\forall j \in [n]$, $x_{aj} = \frac{v_{aj} x_a}{\sum_{h \in [n]} v_{ah}}$ and $x_{bj} = \frac{v_{bj} x_b}{\sum_{h \in [n]} v_{bh}}$ can establish a NE.
    
    For the second case, we consider the new game after merging $n$ battlefields into two battlefields.
    Let $\hat{\bm{x}}_a = (\hat{x}_{a1}, \hat{x}_{a2})$ and $\hat{\bm{x}}_b = (\hat{x}_{b1}, \hat{x}_{b2})$ denote the leader's optimal commitment strategy and the follower's best response respectively.
    By Theorem \ref{sec4-theorem-optimal-commitment}, we can get the leader's optimal commitment, that is
    \begin{equation*}
        \begin{cases}
            \hat{x}_{a1} = \frac{t_1}{t_1 + t_2} x_a;  \\
            \hat{x}_{a2} = \frac{t_2}{t_1 + t_2} x_a.
        \end{cases}
    \end{equation*}
    If $\hat{\bm{x}}_a = (\hat{x}_{a1}, \hat{x}_{a2})$ is also a leader's Nash equilibrium strategy, we have
    \begin{equation}\label{sec5-equation-thm-cal-1}
        \frac{\hat{x}_{a1} v_{bM}}{\hat{x}_{b1} v_{aM}} = \frac{\hat{x}_{a2} v_{b\overline{M}}}{\hat{x}_{b2} v_{a\overline{M}}}.
    \end{equation}
    By Lemma \ref{sec2-lemma-br}, we have $\hat{x}_{bj} = \frac{\left( \hat{x}_{aj} v_{bM} \right)^{\frac{1}{2}} (x_b + x_a)}{\left( \hat{x}_{a1} v_{bM} \right)^{\frac{1}{2}} + \left( \hat{x}_{a2} v_{b\overline{M}} \right)^{\frac{1}{2}}} - \hat{x}_{aj}$, and we substitute it into Equation \eqref{sec5-equation-thm-cal-1}. After some calculations we get
    \begin{equation}\label{sec5-equation-thm-cal-2}
        \frac{x_a}{x_b} = \frac{t_4(M)}{t_3(M)} = f(v_{aT}, v_{bT}, v_{a\overline{T}}, v_{b\overline{T}})
    \end{equation}
    Therefore, we can know $\hat{x}_{a1} > 0$, $\hat{x}_{a2} > 0$, and Equation \eqref{sec5-equation-thm-cal-1} holds when Equation \eqref{sec5-equation-thm-cal-2} is established, and $\hat{x}_{b1} > 0$, $\hat{x}_{b2} > 0$.

    Finally, we need to verify that when Equation \eqref{sec5-equation-thm-cal-2} is established, $\hat{\bm{x}}_a = (\hat{x}_{a1}, \hat{x}_{a2})$ we solved is indeed optimal. By Lemma \ref{sec2-lemma-br}, we know that there are two cases in the follower's best response: (1) positive resources are invested in all battlefields in the follower's best response, and (2) there is one battlefield with no resources invested, and another battlefield with all resources invested in the follower's best response. Without loss of generality, let $(\bm{x'}_{a1}, \bm{x'}_{a2})$ be the another commitment such that $\bm{x'}_{b1} = x_b$ and $\bm{x'}_{b2} = 0$. Because of $\bm{x'}_{b2} = 0$, the leader must invest more resources in the battlefield 2. Thus, we have $x'_{a2} > \hat{x}_{a2}$ and $x'_{a1} < \hat{x}_{a1}$. By Equation \eqref{rewrite-optimal-commitment}, we can calculate the partial derivative of $u_a$ with respect to $\bm{x}_a$, that is
    \begin{equation}
        \frac{\partial u_a}{\partial x_{a1}} = \frac{1}{2(x_a + x_b)} \left( 2v_{aM} + \left(\frac{v_{aM} v_{b\overline{M}}^\frac{1}{2}}{v_{bM}^\frac{1}{2}} + \frac{v_{a\overline{M}} v_{bM}^\frac{1}{2}}{v_{b\overline{M}}^\frac{1}{2}}\right) \frac{x_{a2}^\frac{1}{2}}{x_{a1}^\frac{1}{2}} \right).
    \end{equation}
    Then we have $\frac{\partial u_a}{\partial x'_{a1}} > \frac{\partial u_a}{\partial \hat{x}_{a1}}$ and $\frac{\partial u_a}{\partial x'_{a2}} < \frac{\partial u_a}{\partial \hat{x}_{a2}}$. Due to $\frac{\partial u_a}{\partial \hat{x}_{a1}} = \frac{\partial u_a}{\partial \hat{x}_{a2}}$, then $\frac{\partial u_a}{\partial x'_{a1}} > \frac{\partial u_a}{\partial x'_{a2}}$. Therefore, the leader is willing to increase $x'_{a1}$ and decrease $x'_{a2}$. We can conclude that $\bm{x'}_a$ is not a optimal commitment and $\bm{\hat{x}}_a$ is indeed an optimal commitment.

    It is worth noting that when $\frac{v_{aj}}{v_{bj}} = c$, a constant, case (1) is equivalent to case (2), and $t_4(M) = t_3(M) = 0$.
\end{proof}

\section{OMITTED PROOFS FROM SECTION 6}\label{appendix-section6}
{\bf Proof of Theorem \ref{theorem-ratio}}
\begin{proof}
    Using Lemma \ref{lemma-nash-equilibrium}, we can compute the leader's utility in the Nash equilibrium, that is $u_a^{NE} = \frac{\mu^*}{\mu^* + 1} + \frac{\mu^* v_{a1}}{\mu^* v_{a1} + v_{b1}} v_{a1}$, where $\mu^*$ is a positive solution to $f(\mu) = 0$ defined in Lemma \ref{lemma-nash-equilibrium} and $\mu^* \in [z, \frac{v_{b1}z}{v_{a1}}]$.
    Observe that $u_a^{NE}$ is a monotonically increasing function about $\mu^*$.
    Therefore, we have $\frac{x_av_{a1}^2}{x_av_{a1}+v_{b1}} + \frac{x_a}{x_a+1} \leq u_a^{NE} \leq \frac{v_{b1}x_av_{a1}}{v_{b1}x_a+v_{b1}} + \frac{v_{b1}x_a}{v_{b1}x_a+v_{a1}}$.
    Let $\bm{x}_a$ denote the leader's commitment, and $\bm{x}_b$ denote the follower's best response.
    Then, we can derive that $u_a(\bm{x}_a, \bm{x}_b) \leq u_a(\hat{\bm{x}}_a, \hat{\bm{x}}_b) \leq v_{a1} + 1$, in which the second inequality sign is obtained by allowing the leader to achieve full utility.
    Therefore, we have the following inequalities,
    \begin{align*}
        \frac{u_a(\bm{x}_a, \bm{x}_b)}{\frac{v_{b1}x_av_{a1}}{v_{b1}x_a+v_{b1}} + \frac{v_{b1}x_a}{v_{b1}x_a+v_{a1}}} \leq \frac{u_a(\bm{x}_a, \bm{x}_b)}{u_a^{NE}} \leq \frac{u_a(\hat{\bm{x}}_a, \hat{\bm{x}}_b)}{u_a^{NE}} \leq \frac{v_{a1} + 1}{\frac{x_av_{a1}^2}{x_av_{a1}+v_{b1}} + \frac{x_a}{x_a+1}}.
    \end{align*}
    Consider a certain commitment $\bm{x}_a = (\frac{v_{b1}}{v_{b1} + 1} x_a, \frac{1}{v_{b1} + 1} x_a)$, we have $\bm{x}_b = (\frac{v_{b1}}{v_{b1} + 1}, \frac{1}{v_{b1} + 1})$.
    In this commitment, the leader's utility is
    \begin{align*}
        u_a(\bm{x}_a, \bm{x}_b) = \frac{\frac{v_{b1}}{v_{b1} + 1} x_a}{\frac{v_{b1}}{v_{b1} + 1} x_a + \frac{v_{b1}}{v_{b1} + 1}} v_{a1} + \frac{\frac{1}{v_{b1} + 1} x_a}{\frac{1}{v_{b1} + 1} x_a + \frac{1}{v_{b1} + 1}} = \frac{x_av_{a1} + x_a}{x_a + 1}.
    \end{align*}
    Substituting $u_a(\bm{x}_a, \bm{x}_b)$ into the inequality above yields
    \begin{align*}
        \frac{\frac{x_av_{a1} + x_a}{x_a + 1}}{\frac{v_{b1}x_av_{a1}}{v_{b1}x_a+v_{b1}} + \frac{v_{b1}x_a}{v_{b1}x_a+v_{a1}}} \leq \frac{\frac{x_av_{a1} + x_a}{x_a + 1}}{u_a^{NE}} \leq \frac{u_a(\hat{\bm{x}}_a, \hat{\bm{x}}_b)}{u_a^{NE}} \leq \frac{v_{a1} + 1}{\frac{x_av_{a1}^2}{x_av_{a1}+v_{b1}} + \frac{x_a}{x_a+1}}.
    \end{align*}
    After simplification, we obtain
    \begin{align*}
        \frac{v_{a1} + 1}{v_{a1} + \frac{v_{b1}(x_a+1)}{v_{b1}x_a+v_{a1}}} \leq \frac{u_a(\hat{\bm{x}}_a, \hat{\bm{x}}_b)}{u_a^{NE}} \leq \frac{v_{a1} + 1}{\frac{x_av_{a1}^2}{x_av_{a1}+v_{b1}} + \frac{x_a}{x_a+1}}.
    \end{align*}
    On the one hand, it is easy to obtain that $\lim_{z \rightarrow \infty} \frac{v_{a1} + 1}{\frac{x_av_{a1}^2}{x_av_{a1}+v_{b1}} + \frac{x_a}{x_a+1}} = 1$.
    On the other hand, we let $v_{a1} = x_a$ and $v_{b1} = v_{a1}^2$, we have $\lim_{x_a \rightarrow 0} \frac{v_{a1} + 1}{v_{a1} + \frac{v_{b1}(x_a+1)}{v_{b1}x_a+v_{a1}}} = \infty$.
\end{proof}

\begin{lemma}\label{lemma-nash-equilibrium}
    [Li and Zheng (2022)] In a two-player Lottery Colonel Blotto game, assume that $\frac{v_{b1}}{v_{a1}} \leq \frac{v_{b2}}{v_{a2}} \leq \cdots \leq \frac{v_{bn}}{v_{an}}$. Let
    $$
    f(\mu) = \sum_{h=1}^n \left[ v_{ah}\left(\frac{v_{bh}}{v_{ah}}\right) \mu \left(\mu-\frac{v_{bh}}{v_{ah}} \frac{x_a}{x_b}\right) \cdot \prod_{j \neq h}\left( \mu + \frac{v_{bj}}{v_{aj}} \right)^2 \right],
    $$
    and $\mu^*$ is a positive solution to $f(\mu) = 0$.
    Then, $(\bm{x^*}_a, \bm{x^*}_b)$ is a Nash equilibrium of the game if and only if
    \begin{align*}
        x_{ah}^* &= \mu^* \cdot \frac{v_{bh} \mu^* / \left( \mu^* + \left( \frac{v_{bh}}{v_{ah}}\right) \right )^2}{\sum_{j=1}^n v_{bj} \left( \frac{v_{bj}}{v_{aj}} \right) \mu^* / \left( \mu^* + \left( \frac{v_{bj}}{v_{aj}} \right) \right)^2} \cdot x_b, \\
        x_{bh}^* &= \frac{v_{bh} \left(\frac{v_{bh}}{v_{ah}} \right) \mu^* / \left( \mu^* + \left( \frac{v_{bh}}{v_{ah}}\right) \right )^2}{\sum_{j=1}^n v_{bj} \left( \frac{v_{bj}}{v_{aj}} \right) \mu^* / \left( \mu^* + \left( \frac{v_{bj}}{v_{aj}} \right) \right)^2} \cdot x_b,
    \end{align*}
    and $\mu^* \in [\min_h \frac{v_{bh}}{v_{ah}} \frac{x_a}{x_b}, \max_h \frac{v_{bh}}{v_{ah}}\frac{x_a}{x_b}]$.
\end{lemma}

\end{document}